\documentclass[reqno,12pt,letterpaper]{amsart}
\usepackage[usenames,dvipsnames]{color}
\usepackage{amsmath,amssymb,amsthm,graphicx,mathrsfs,url}
\usepackage[alphabetic,nobysame]{amsrefs}
\usepackage[colorlinks=true,linkcolor=blue,citecolor=blue]{hyperref}
\usepackage{amsxtra}
\usepackage{subcaption}
\usepackage{soul}
\usepackage{aas_macros}
\usepackage{autonum}
\usepackage{pgfplots}
\usepackage{amsmath,amssymb,amsthm,graphicx,mathrsfs,url}
\usepackage{slashed}
\usepackage{tikz}
\usetikzlibrary{shadings}
\newtheorem{theo}{Theorem}
\newtheorem{prop}{Proposition}[section]	
\usepackage{fixltx2e} 
\usepackage{floatrow}
\usepackage{bm}
\usepackage{cancel}

\usepackage[top=1in,left=1in,right=1in,bottom=1in]{geometry}
\newtheorem{conj}{Conjecture}

\newtheorem{lemm}[prop]{Lemma}

\theoremstyle{definition}

\newtheorem{rem}{Remark}

\newcommand{\field}[1]{\mathbb{#1}} 

\numberwithin{equation}{section}

\newcommand{\eps}{\varepsilon}

\newcommand{\II}{\mathcal{I}}
\newcommand{\DD}{\mathcal{D}}

\newcommand{\R}{\mathbb{R}}
\newcommand{\C}{\mathbb{C}}
\newcommand{\epsi}{\varepsilon}
\newcommand{\N}{\mathbb{N}}
\newcommand{\p}{\partial}
\newcommand{\te}{\theta}
\newcommand{\RR}{\mathcal{R}}
\newcommand{\Or}{\mathcal{O}}

\newcommand{\matrice}[1]{\left[  \begin{matrix} #1 \end{matrix}\right]}
\newcommand{\systeme}[1]{\left\{  \begin{matrix} #1 \end{matrix}\right.}
\newcommand{\lr}[1]{\left\langle #1 \right\rangle}

\newcommand{\vp}{{\varphi}}
\newcommand{\hh}{{\mathfrak{h}}}
\newcommand{\aaa}{{\mathfrak{a}}}
\newcommand{\SSS}{{\mathcal{S}}}
\newcommand{\VV}{{\mathcal{V}}}
\newcommand{\KK}{{\mathcal{K}}}

\newcommand{\tL}{{\tilde{L}}}
\newcommand{\tSSS}{{\tilde{\SSS}}}
\newcommand{\sgn}{\operatorname{sgn}}
\newcommand{\Hm}{{\boldmath{$(\operatorname{H}_m)$}} }
\newcommand{\Hmm}{{\boldmath{$(\operatorname{H}_{m-1})$}} }

\newcommand{\Hone}{{\boldmath{$(\operatorname{H}_1)$}} }

\newcommand{\UU}{\mathcal U}

\newcommand{\BB}{\mathcal{B}}

\newcommand{\WF}{WF}
\newcommand{\CC}{\mathcal{C}}
\newcommand{\tkappa}{{\tilde{\kappa}}}
\newcommand{\trho}{{\tilde{\rho}}}

\definecolor{dblue}{RGB}{57, 5, 179}

\author{G. Bal}
\address[Guillaume Bal]{University of Chicago, USA.}
\email{guillaumebal@uchicago.edu}

\author{S. Becker}
\address[Simon Becker]{University of Cambridge, United Kingdom.}
\email{simon.becker@damtp.cam.ac.uk}

\author{A. Drouot}
\address[Alexis Drouot]{University of Washington, USA.} 
\email{adrouot@uw.edu}

\author{C. Fermanian Kammerer}
\address[Clotilde Fermanian Kammerer]{Universit\'e Paris Est - Cr\'eteil Val de Marne, France.}
\email{clotilde.fermanian@u-pec.fr}

\author{J. Lu}
\address[Jianfeng Lu]{Duke University, USA.}
\email{jianfeng@math.duke.edu}

\author{A. Watson}
\address[Alexander Watson]{University of Minnesota, USA.}
\email{watso860@umn.edu}

\title{Edge state dynamics along curved interfaces}

\begin{document}

\begin{abstract}
We study the propagation  of wavepackets along weakly curved interfaces between topologically distinct media. Our Hamiltonian is an adiabatic modulation of Dirac operators omnipresent in the topological insulators literature.
%and provides an explicit example of the bulk-edge correspondence and the quantum anomalous Hall effect. 
Using explicit formula{s} for straight edges, we construct a family of solutions that propagates, for long times, unidirectionally and dispersion-free along the curved edge. We illustrate our results through various numerical simulations. 
\vspace{-.8cm}
\end{abstract}

\maketitle

\section{Introduction}

\begin{figure}[b]
\vspace{-5mm}
\floatbox[{\capbeside\thisfloatsetup{capbesideposition={right,center},capbesidewidth=3.7in}}]{figure}[\FBwidth]
{\hspace{-1cm}\caption{\label{fig:5}  Snapshots of the numerically computed dynamical analogue of an edge state -- the solution to \eqref{eq:0r} below. The  interface is $y_2=\tanh(y_1)$ and $\epsi = 10^{-1}$. The state propagates leftwards and dispersion-free along the interface.  
}}
{\begin{tikzpicture}
   \node at (0,0) {\includegraphics[width=7cm, height = 5cm]{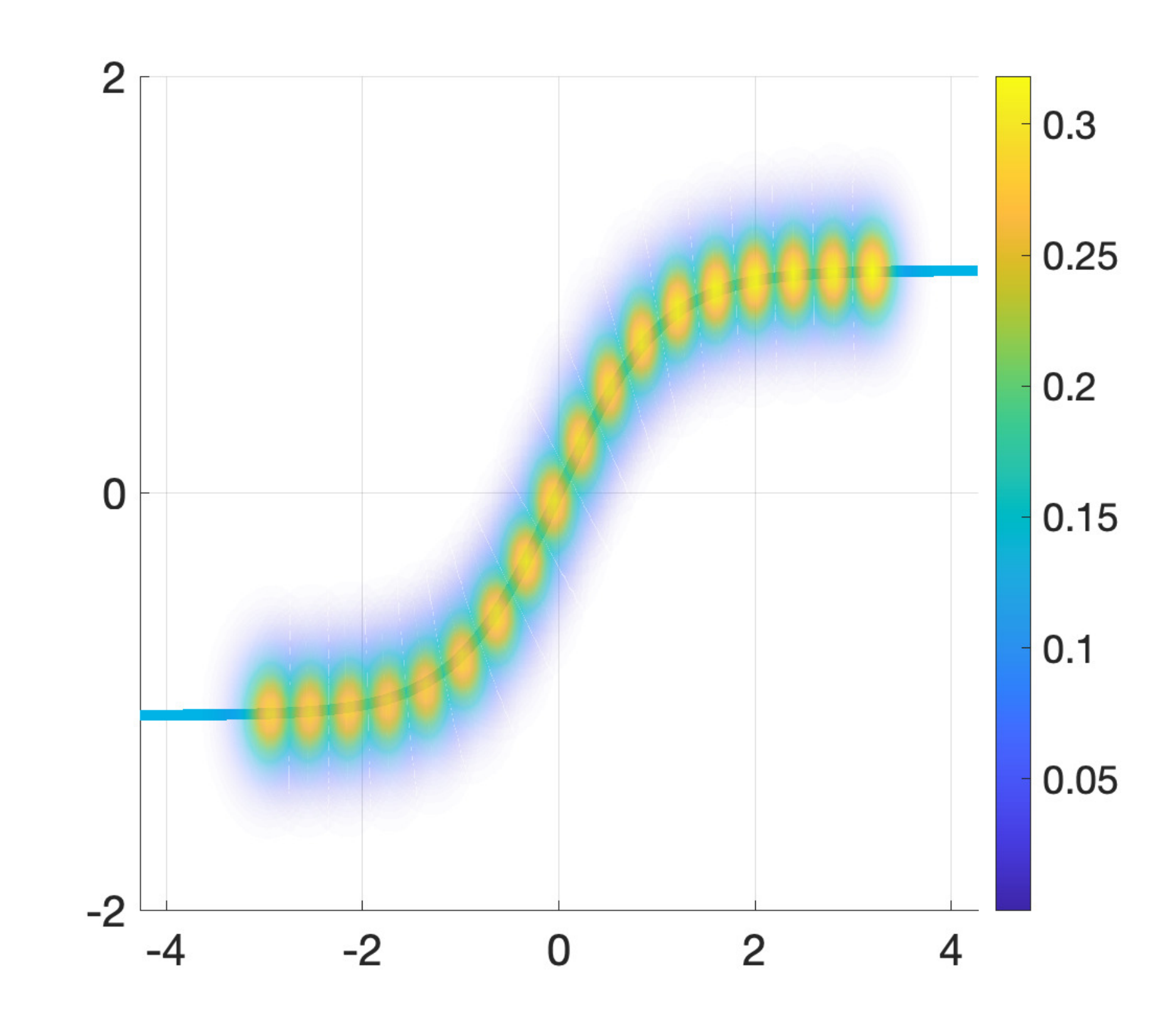}};
 \draw[domain=1.2:-1.2 , smooth, variable=\x, ultra thick,red,->] plot (\x-.2, {1.1*tanh(1.7*\x)+.7 });
\node[red] at (-1.4,1.2) {direction of};
\node[red] at (-1.4,.8) {propagation};
\node[red] at (1.7,.7) {$\uparrow$};
\node[red] at (1.7,.2) {initial};
\node[red] at (1.7,-.2) {state};
 \end{tikzpicture}} 
\end{figure}
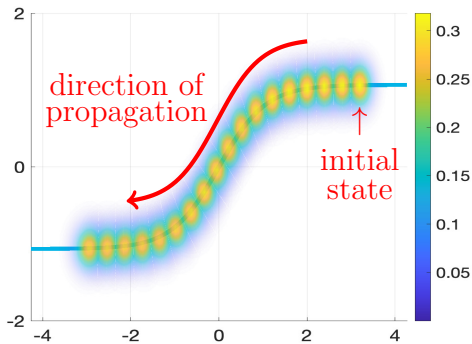

Topological insulators are fascinating materials that are insulating in their bulk but support robust currents along their boundary. From a mathematical point of view, these properties are consequences of the bulk-edge correspondence, an index-like theorem that relates the net conductivity (an analytic index) to the bulk topology (a topological index). For straight interfaces, the currents are explicitly described in terms of edge states: steady waves with ballistic dynamics, confined between regions of distinct topology. 

In this work, we construct dynamical analogues of edge states for curved interfaces. Our model is a Dirac operator
\begin{equation}\label{eq:1b}
H = \matrice{ \kappa(x) & \epsi D_{x_1} - i \epsi D_{x_2} 
\\
\epsi D_{x_1} + i \epsi D_{x_2} & -  \kappa(x) }
 %= \epsi \sigma_1 D_{x_1} + \epsi \sigma_2 D_{x_2} + \sigma_3 \kappa(x),
\end{equation} 
where $D_{x_j}=-i \partial_{x_j}$, $\epsi > 0$ is a small semiclassical parameter and $\kappa$ is a varying mass term. Such Hamiltonians emerge in the effective theory of honeycomb structures  \cite{FLW16,LWZ19,Drouot:19}; more generally they model the generic dynamics of modes propagating along interfaces between topologically distinct insulators \cite{Drouot:21}.

%\alexis{there were paragraphs discussing how \eqref{eq:1b} connects to the adiabatic model we initially considered. I took them off, since even \S\ref{sec:1.3} does not refer to them.}

Under a transversality condition -- $\nabla \kappa(x) \neq 0$ when $\kappa(x) = 0$ -- the set 
\begin{equation}\label{eq:5a}
  \Gamma =  \{ x \in \R^2 : \kappa(x) = 0\}
\end{equation}  
partitions $\R^2$ in regions of distinct local topology -- see \S\ref{sec:1.4} for details. A local interpretation of the bulk-edge correspondence suggests that non-trivial currents emerge along $\Gamma$. This paper develops the underlying quantitative theory: it provides detailed information on the associated quantum states, such as  their speed and profile. 

\iffalse

For small $\epsi \ll 1$, the rescaling $x \mapsto \epsi x$ of \eqref{eq:1b} produces the adiabatic Hamiltonian (i.e. with slowly varying coefficients):
\begin{equation}\label{eq:0t}
 \matrice{ \kappa(\epsi x) &  D_{x_1} - i D_{x_2} 
\\
D_{x_1} + iD_{x_2} & -  \kappa(\epsi x) }.
\end{equation}
Adiabaticity refers to the slow variation scale ($\epsi^{-1} \gg 1$) of the coefficients of \eqref{eq:0t}. In this regime, we can define for \eqref{eq:0t} notions of local topology via the sign of $\kappa$, see \S\ref{sec:1.3}. Hence, under a transversality condition on $\kappa$ -- $\nabla \kappa(y) \neq 0$ when $\kappa(y) = 0$ -- the set 
\begin{equation}\label{eq:5a}
  \Gamma =  \{ y \in \R^2 : \kappa(y) = 0\}
\end{equation}  
partitions $\R^2$ in regions of distinct local topology. 
Local interpretations of conductivity and of the bulk-edge correspondence suggest then the existence of non-trivial currents along $\Gamma$. As a qualitative statement, it does not provide detailed information on the channels of conductivity (such as their speed or profile); rather it relates the bulk topological difference to the net conductivity. For adiabatic operators such as \eqref{eq:1b}, this series of papers aims to develop the underlying quantitative theory.

\fi

Specifically, we exploit the explicit structure of edge states available when $\kappa(x) = a_1 x_1 + a_2 x_2$ to construct an infinite-dimensional family of nearly steady solutions to $(\epsi D_t + H) \psi = 0$, in the limit $\epsi \rightarrow 0$. These emerge as the natural channels of conductivity: for long times, they propagate unidirectionally and coherently along $\Gamma$. We show that the curvature of $\Gamma$ plays a key role in limiting the lifetime of these solutions. We illustrate our results via various numerical simulations.

\subsection{Simplified main result}   
Throughout the paper, we assume that $\kappa$ and all its derivatives are bounded: $\kappa \in C^\infty_b(\R^2)$. In this introduction, we require moreover that 
\begin{equation}\label{eq:3r}
    y \in \Gamma \ \ \ \Rightarrow \ \ \ \big| \nabla \kappa(y) \big| = 1.
\end{equation}
This allows us to state a simplified version (Theorem \ref{thm:1}) of our main result (Theorem \ref{thm:2}). In \S\ref{sec:3}, we replace \eqref{eq:3r} by the more general transversality condition  \eqref{eq:7z}.

Fix $y_0 \in \Gamma = \kappa^{-1}(0)$ and define $y_t$ by the ODE
\begin{equation}\label{def:yt}
    \dot{y_t} = \nabla \kappa(y_t)^\perp,
\end{equation}
where $\nabla \kappa(y)^\perp$ denotes the $\pi/2$-counterclockwise rotation of $\nabla \kappa(y)$. Under \eqref{eq:3r}, $y_t$ is a unit speed parametrization of $\Gamma$. We let $\te_t$ be the angle between the tangent to $\Gamma$ at $y_t$ and the $x$-axis -- see Figure \ref{fig:4}. We use the notation $\lr{t} = (1+|t|^2)^{1/2}$.
%\alexis{We should have a uniform notation: $\Gamma$ or $\kappa^{-1}(0)$?}

\begin{figure}[b]
\floatbox[{\capbeside\thisfloatsetup{capbesideposition={right,center},capbesidewidth=2.8in}}]{figure}[\FBwidth]
{\hspace{-1cm}\caption{\label{fig:4} 
Schematic plot of an interface $\Gamma = \kappa^{-1}(0)$ between topologically distinct regions, together with $y_t$ and $\te_t$.
}}
{\begin{tikzpicture}
   \node at (0,0) {\includegraphics[]{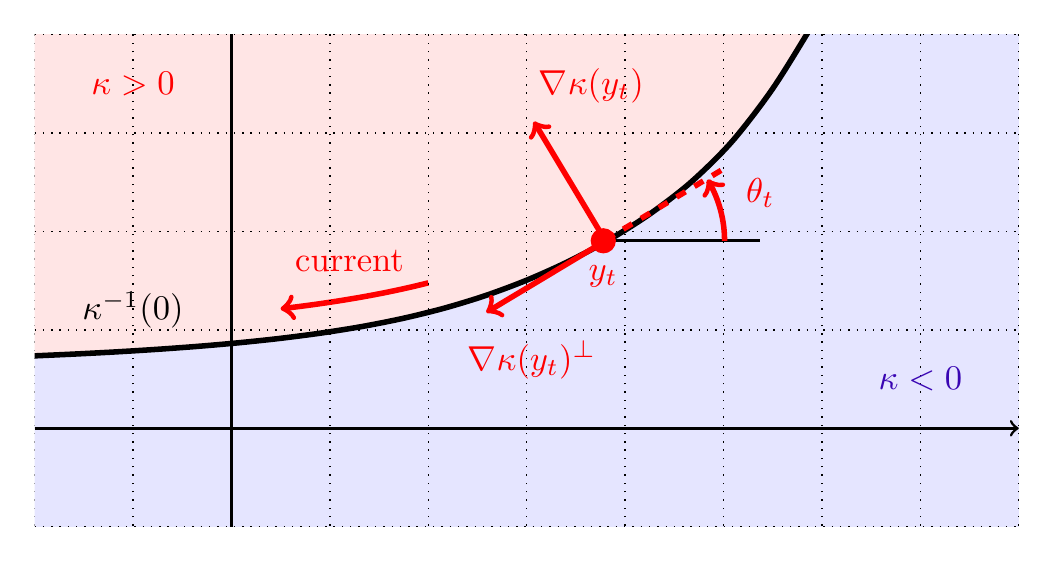}};
  \end{tikzpicture}}
\end{figure}

\begin{theo}\label{thm:1} Let $\kappa \in C^\infty_b(\R^2)$ satisfy \eqref{eq:3r} and $y_t$, $\te_t$ as above. The solution to  
\begin{equation}\label{eq:0r}
 (\eps D_t+H)\Psi_t=0, \ \ \ \    \Psi_0(x) = \dfrac{1}{\sqrt{\epsi}} \cdot \exp\left( - \dfrac{(x-y_0)^2}{2\epsi}\right) \matrice{ e^{-i\te_0/2} \\ -e^{i\te_0/2} }
\end{equation}
satisfies, uniformly for $\epsi \in (0,1]$ and $t > 0$:
\begin{equation}\label{eq:3s}
\Psi_t(x) = \dfrac{1}{\sqrt{\epsi}} \cdot \exp\left( - \dfrac{(x-y_t)^2}{2\epsi}\right) \matrice{ e^{-i\te_t/2} \\ -e^{i\te_t/2} } + \Or_{L^2}\left(\epsi^{1/2} \lr{t}\right).
\end{equation}  
\end{theo}

The initial data \eqref{eq:0r} is a  
Gaussian concentrated at $y_0$.  Theorem \ref{thm:1} shows that the generated solution remains (at leading order, for times $t \ll \epsi^{-1/2}$) a Gaussian, concentrated now at $y_t$. This identifies $t \mapsto y_t$ as an exotic quantum trajectory: it is not predicted by the standard results on propagation of semiclassical singularities. See \S\ref{sec:1.4} for a semiclassical discussion.

%The solution to \eqref{eq:0r} is well-described by the Gaussian state in \eqref{eq:3s} for times $t \ll \epsi^{-1/2}$. 
If $\Gamma$ is not asymptotically straight -- for instance if it is a loop -- numerical computations confirm that the Gaussian state approximation becomes less and less accurate, see  Figure \ref{fig:7}. In contrast, if $\Gamma$ is asymptotically straight -- as in e.g. the $\tanh$-like interface of Figure \ref{fig:5} -- the Gaussian state approximation can work for longer times, see Theorem \ref{thm:3}.

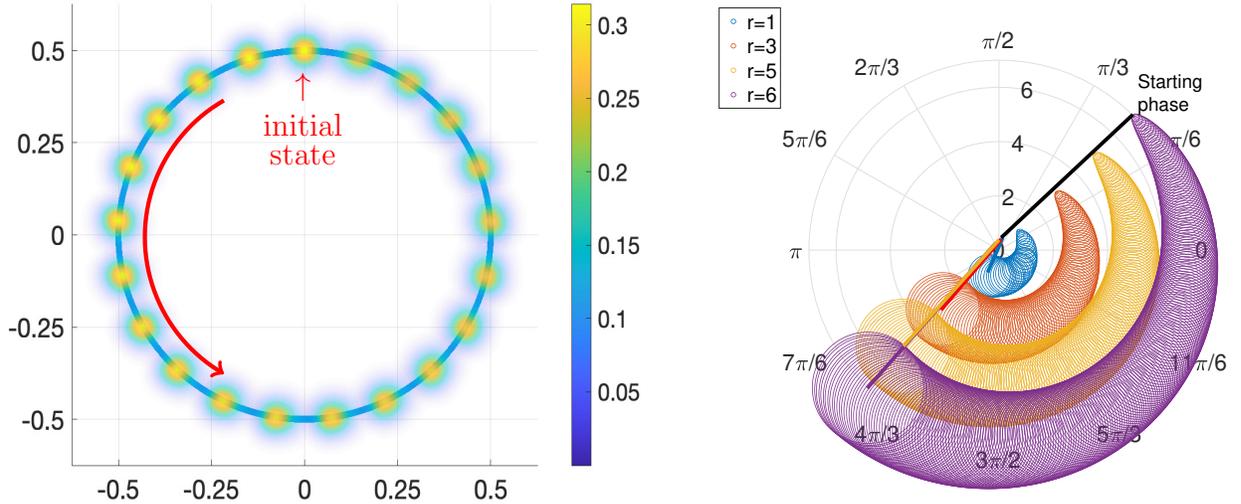
\begin{figure}[b]
{\begin{tikzpicture}
   \node at (-4.5,0) {\includegraphics[width = 9cm, height = 7.5cm]{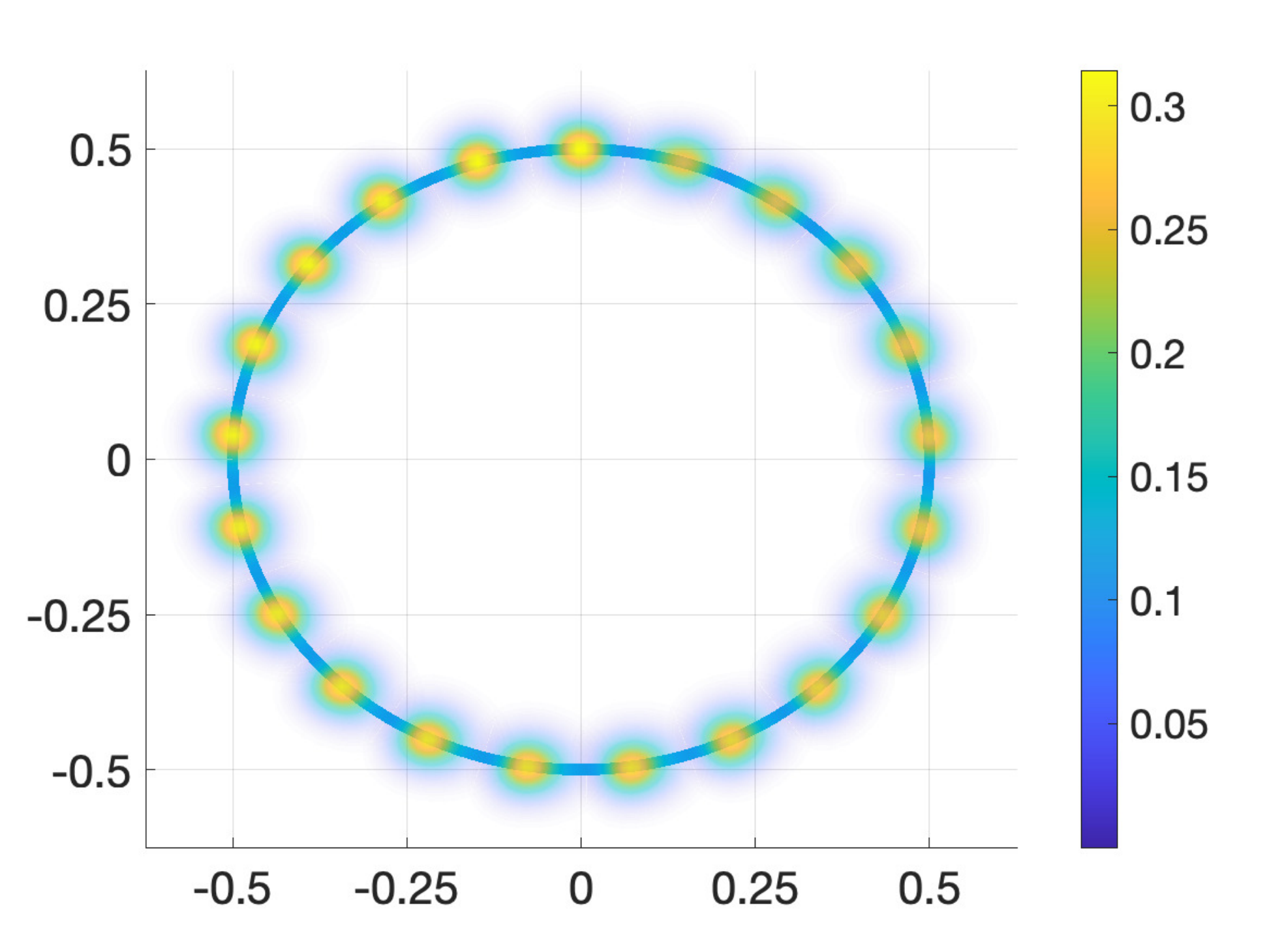}};
  \node at (4,0) {\includegraphics[width = 6.75cm]{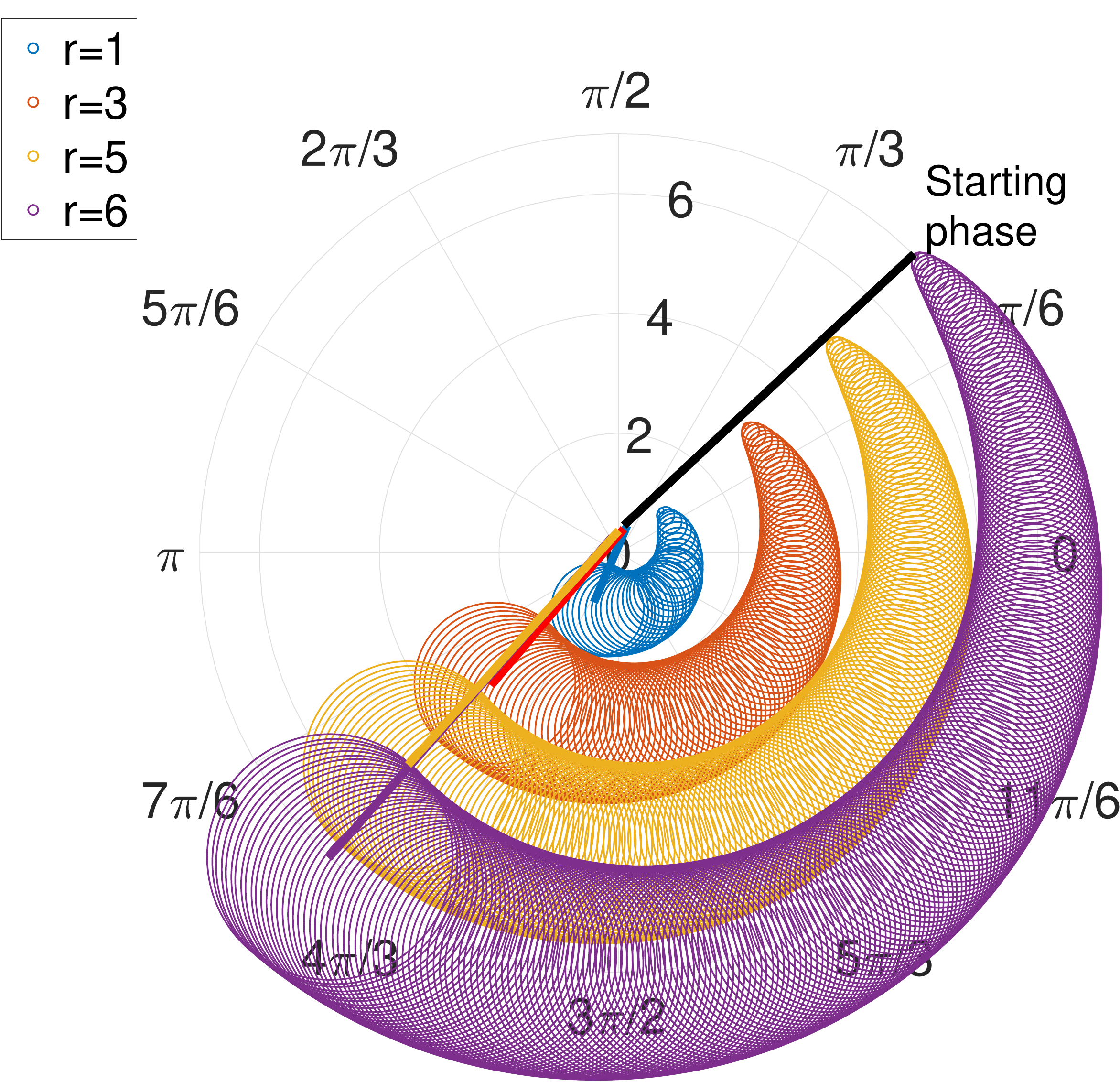}};

 \draw[domain=120:240 , smooth, variable=\x, ultra thick,red,->] plot ({2.1*cos(\x)-4.9}, {2.1*sin(\x)+.1});

\begin{scope}[shift={(-4.9,1.4)}]
\node[red] at (0,.7) {$\uparrow$};
\node[red] at (0,.2) {initial};
\node[red] at (0,-.2) {state};
\end{scope}

  \end{tikzpicture}}
  \caption{\label{fig:7}
  Left: numerical solution to $(\epsi D_t + H) \Psi_t = 0$ with Gaussian initial state for a circular interface with $\epsi = 10^{-2}$ and radius one. The 
    trajectory $y_t$ undergoes curvature effects for all times. This explains a dispersion stronger than for a $\tanh$-type interface. See also Figure \ref{fig:8} and Theorem \ref{thm:3}. 
Right: evolution of the phase of the first coordinate of the numerical solution for each snapshot -- corresponding to $-\te_t/2$ -- for different radii of the circle-interface. After a full revolution, the numerical phase difference is about $-\pi$, matching the theoretical prediction $-2\pi / 2 = -\pi$. This phase shift interprets as a Berry phase arising from adiabatically varying the parameter $\theta$ in the effective leading order operator $H_{\theta,r}$ \eqref{eq:Hte} from $0$ to $2 \pi$. 
} 
\end{figure}

We refer to Theorem \ref{thm:2} for a more general version of Theorem \ref{thm:1}. It constructs an infinite dimensional family of solutions to $(\epsi D_t + H) \Psi_t = 0$ with the same qualitative features as \eqref{eq:3s}: coherent states propagating unidirectionally, at unit speed and without dispersion, along $\Gamma$. Our motivation, explained in \S\ref{sec:1.3} and \S\ref{sec:1.4} below, is two-fold:
\begin{itemize}
    \item Identify dynamical analogues of topological edge states along bent interfaces;
    \item Study a semiclassical system whose matrix-valued symbol has repeated eigenvalues. \end{itemize}

\subsection{Numerical simulations} We illustrate our results with numerical simulations of the Dirac equation with a Gaussian initial data, for various types of interfaces. The corresponding pictures are snapshots of the dynamics, with the interface marked as a light blue curve.
\begin{itemize}
    \item Figure \ref{fig:5} and \ref{fig:7} are numerical confirmations of Theorem \ref{thm:1} for tanh-type and circle interfaces, respectively. Figure \ref{fig:7} also verifies that the phase shift after one revolution equals $2\pi/2 = \pi$.
    \item Figure \ref{fig:6} shows the evolution of other Gaussian states for $\tanh$-type interfaces. The initial data are concentrated like \eqref{eq:0r} but carried by a different vector. If this vector is orthogonal to that in \eqref{eq:0r}, the coherence is immediately lost. See Conjecture \ref{conj:1}.
    \item When the more general transversality condition \eqref{eq:7z} holds instead of \eqref{eq:3r}, the propagation is coherent in a relaxed sense.  Figure \ref{fig:1} -- a straight interface but a non-linear domain wall -- numerically validates Theorem \ref{thm:2}. 
%    \item Figure \ref{fig:my_label} shows how the curvature of the interface affects corrections to the Gaussian state. See \S\ref{sec:4} for a discussion relating curvature to lifetime.  
    \item Figure \ref{fig:8} illustrates the limits of the dynamical analogues of edge states: for instance, they do not propagate around sharp corners. 
\end{itemize}

%\begin{rem} 
We use a Crank-Nicholson scheme to approximate the unitary group $e^{-itH}$, with Fourier spectral spatial discretization. 
The Matlab code containing the parameters used to obtain our figures can be found on GitHub.\footnote{\href{https://github.com/slb2604/Semiclassical-edge-states}{https://github.com/slb2604/Semiclassical-edge-states}}
%\end{rem}

\subsection{Physical motivations.} 

The Dirac equation appears in a wide variety of physical applications. Beyond its original role in the description of relativistic particles, %(see e.g. \cite{TH} for a review)
it has emerged as a dominant model in the analysis of topological phases of matter \cite{VO,WI}.
%which studies how phase-space topology affects transport properties. 
The relativistic Dirac operator ($\kappa=0$ in our model) displays a generic band crossing; in contrast, adding a mass term opens an energy gap. In our model, the interface is the transition between the two insulating phases $\kappa<0$ and $\kappa>0$. These two phases happen to have different topological signatures; this generates unidirectional propagation along the interface.

This asymmetric transport is at the core of most physical applications in the fields of topological insulators  and topological superconductors \cite{BE,VO}. It is the physical manifestation of the quantum Hall effect \cite{BES94,ASS90} and its non-magnetic analogues \cite{C13,H88,JS20,HIA19,LD20,Sp18}. It also finds numerous applications in fields such as photonics, acoustics, and fluid mechanics \cite{LJS,PBMM,RH,RPZ,GJT}.  Broadly speaking, Dirac-type equations often offer the simplest continuum (macroscopic) description of transport in a narrow energy band near the band crossing \cite{BE,FC,VO}.

\subsection{Local topological indices and asymmetric transport.}\label{sec:1.3}  Strikingly, transport at interfaces between distinct topological environments is both asymmetric (a net overall flux propagates in a prescribed direction) and quantized.  We discuss here a theory of topological phases that interprets locally the state \eqref{eq:3s} in a topological way. We stress that this interpretation:
\begin{itemize}
    \item is valid only in the semiclassical regime $\epsi \ll 1$;
    \item is local: our construction works for all $\kappa$, even though in some scenarios $H$ is topologically trivial (for instance when $\Gamma$ is a closed curve).
\end{itemize}

These considerations use the leading-order approximation $H_y$ of $H$ at a point $y \in \R^2$:
\begin{align}
    H_y = \matrice{\kappa(y) & \epsi D_{x_1} - i \epsi D_{x_2} \\ \epsi D_{x_1} + i \epsi D_{x_2} & -\kappa(y) }, \ \ \ \ y \notin \Gamma;
    \\
    H_y = \matrice{-v_y^\perp \cdot (x-y) & \epsi D_{x_1} - i \epsi D_{x_2} \\ \epsi D_{x_1} + i \epsi D_{x_2} & v_y^\perp \cdot (x-y) }, \ \ \ \ y \in \Gamma,
\end{align}
where $v_y= \nabla\kappa(y)^\perp$ is tangent  to $\Gamma$ at $y$.
These emerge by replacing $\kappa(x)$ in \eqref{eq:1b} by its leading-order development at $y$: $\kappa(x) \simeq \kappa(y)$ if $y \notin \Gamma$ and $\kappa(x) \simeq \nabla\kappa(y) \cdot (x-y)$ if $y \in \Gamma$. These approximations are reasonable for $|x-y| = O(\epsi^{1/2})$: the scale of localization of \eqref{eq:3s}.

We observe that $H_y$ has a spectral gap near energy $0$ (i.e. it is an insulator) if and only if $y \notin \Gamma$. This identifies $\Gamma$ as the natural channel for conduction of energy. Following \cite{EG02,Kellendonk}, we measure the local conductivity at $y \in \Gamma$ via:
\begin{equation}\label{eq:0l}
   \II(H,y) =  \operatorname{Tr}_{L^2}\Big( i\big[H_y,f(v_y \cdot x)\big] g'\big( H_y \big) \Big), \ \ \ \ \ y \in \Gamma, 
\end{equation}
where $f$ and $g$ are smooth real functions  
 increasing from $0$ to $1$ with $f'$ and $g'$ compactly supported. 
Formally,
\begin{equation}\label{eq:9g}
    \II(H,y) = \dfrac{d}{dt} \operatorname{Tr}_{L^2} \left( e^{itH_y} f(v_y \cdot x) g'\big( H_y \big) e^{-itH_y} \right).
\end{equation}
Looking at $g'$ as a density of probability, $f(v_y \cdot x) g'\big( H_y \big)$ measures the probability of a quantum particle to lie in the half-plane $\{v_y \cdot x > 0\}$, per unit energy. Taking the trace in \eqref{eq:9g} corresponds to summing over all states. Hence $\II(H,y)$ describes the overall flux moving in the direction of $v_y$, per unit time and energy, at equilibrium.

It turns out that $2\pi \cdot \II(H,y) = 1$, see \cite{B19b} and Remark \ref{rem:1} below. This means that the evolution according to $H_y$ comes with a current propagating in the direction of $v_y$. Since $v_y$ is tangent to $\Gamma$ at $y$, $\Gamma$ emerges intuitively as a natural charge-carrier for $H$. Theorem \ref{thm:1} confirms these heuristics: in the regime $\epsi \rightarrow 0$, we construct a current propagating along $\Gamma$, with explicit  speed and profile.

The quantity \eqref{eq:0l} relates to bulk topological invariants via a universal principle: the bulk-edge correspondence \cite{Hatsugai,GP,prodan2016bulk,B20,drouot2020microlocal}. Following the physics literature \cite{H88,HIA19}, we  define a bulk index for $H_y$:
\begin{equation}\label{eq:0k}
 \BB(H,y) = \dfrac{\sgn \bigl( \kappa(y) \bigr)}{2}, \ \ \ \  y \notin \Gamma.
\end{equation}
When $H$ emerges as an effective Hamiltonian (for instance in graphene), $\BB(H,y)$ corresponds to the integrated Berry curvature near one of the Dirac point momentum, hence as part of the overall Chern integer \cite{Drouot:19b}. Direct interpretations of \eqref{eq:0k} as a Chern number include regularization of Dirac operator \cite{B19b} and more general bulk-difference invariant \cite{B20}. We refer to \eqref{eq:0k} as the local bulk index. It can also be defined by spatially truncating physical space formulas for the global Chern number \cite{2006Kitaev,2011BiancoResta,prodan2016bulk}; or via the spectral localizer \cite{2015Loring,LoringSchulz-Baldes2019}.

Since $\nabla\kappa$ points from negative to positive-index regions, we have for $y \in \Gamma$ and $\delta > 0$  sufficiently small:
\begin{equation}
  1 = 2\pi \cdot \II(H,y) =   \BB\big(H,y + \delta \nabla \kappa(y)\big) - \BB\big(H,y - \delta \nabla \kappa(y)\big).
\end{equation}
This is a local version of the bulk-edge correspondence: the local conductivity at $y$ is the difference between the local bulk indices across the interface.

The quantity $2\pi \cdot \II(H,y)$ counts currents algebraically according to their direction of propagation. It is independent of $y$ and stable against large perturbations of $H$; see, e.g. \cite{B19b,B20} and \cite{prodan2016bulk} for similar models. This explains its practical significance: even in the presence of strong perturbations or Anderson localization, there is always $2\pi \cdot \II(H,y) = 1$ more current propagating in the direction of $v_y$ rather than $-v_y$ \cite{B19a,prodan2016bulk}.
%For the unperturbed Hamiltonian $H_y$ ($y \in \Gamma$), and locally for the Hamiltonian $H$, the asymmetric transport is asymptotically described by \eqref{eq:3s}. 
This clarifies the \textit{local} topological nature of the quantum state \eqref{eq:3s}. Let us stress again that our results hold locally in time: \eqref{eq:0l} is spectral in nature, describing an equilibrium, while  \eqref{eq:3s} is relevant for (long, but only transient) times $t \ll \epsi^{-1/2}$.

\subsection{Connection with semiclassical analysis}\label{sec:1.4} What makes the solution \eqref{eq:3s} special? The answer lies in semiclassical territory. In summary (with details provided below): if $\CC = \Gamma\times \{0\} \subset \R^2 \times \R^2$, then for times $t \ll \epsi^{-1/2}$:
\begin{itemize}
     \item[(i)] States initially microlocalized at $(y_0, \xi_0) \notin \CC$ come in pairs propagating in opposite directions;
     \item[(ii)] States initially microlocalized at $(y_0,\xi_0) \in \CC$ (i.e. like \eqref{eq:0r}, with a potentially different 2-vector) seem to either propagate non-dispersively in the direction of $\nabla \kappa^\perp$, or to disperse; see Figure \ref{fig:6}  and Conjecture \ref{conj:1}. 
\end{itemize}
This suggests that $\Gamma$ -- more precisely, its phase-space lift $\CC$ -- is the relevant channel for asymmetric propagation. 

We now provide a detailed account. We start by writing $H = \hh(x, \epsi D_x)$, where  
\begin{equation}\label{eq:5b}
    \hh(x,\xi) = \matrice{ \kappa(x) & \xi_1 - i\xi_2 \\ \xi_1 + i \xi_2 & -\kappa(x) }.
\end{equation}

\noindent Theorem \ref{thm:1} constructs solutions to 
$\big(\epsi D_t + \hh(x,\epsi D_x) \big) \phi_t = 0$   
for the data 
\[\phi_0(x) = \dfrac{1}{\sqrt{\epsi}} \cdot e^{\frac{i}{\epsi} x \xi_0} \  a \left( \dfrac{x-x_0}{\sqrt{\epsi}} \right), \ \ \ \ a \in \SSS\big(\R^2,\C^2\big)
\]
where $(x_0,\xi_0)$ belongs to the set $\CC$ defined by
\begin{equation}
    \CC = \big\{ (x,\xi) :  \ \kappa(x) = 0, \ \xi=0 \big\} \subset \R^4.
\end{equation} 
The function $\phi_0$ is known in the literature as a semiclassical wavepacket \cite{Combescure_Robert} with wavefront set $\WF_\epsi(\phi_0) = \{(x_0,\xi_0)\}$ -- see \cite[\S8.4]{Zw} for definitions and properties of wavefronts. The set $\CC$ corresponds to semiclassical eigenvalue crossings of $\hh(x,\xi)$: when $(x,\xi) \in \CC$, $\hh(x,\xi)$ has two degenerate eigenvalues. The systematic study of such semiclassical systems is a delicate problem. In the context of the Landau--Zener effect, which corresponds to a varying crossing energy, we refer to \cite{Col04} for a derivation of local normal forms, and to \cite{Hag94} for an explicit description of the transition.

This paper focuses on the dynamics of wavepackets localized along $\CC$ (note that the crossing energy is constant, equal to $0$).
One could have likewise studied the dynamics of wavepackets semiclassically concentrated at points $(x_0, \xi_0) \notin  \CC$. This is actually a much more standard problem because the eigenvalues of $\hh(x_0,\xi_0)$ are distinct: they are $\pm \lambda(x_0,\xi_0)$, where
\begin{equation}
\lambda(x,\xi) = \sqrt{\kappa(x)^2 + \xi_1^2 +\xi_2^2};
\end{equation}
we note that $\lambda$ does not vanish away from $\CC$. We diagonalize $\hh(x,\xi)$ for $(x,\xi)$ near $(x_0,\xi_0)$:
\begin{equation}
    \hh(x,\xi) = \mathfrak{U}(x,\xi) \matrice{-\lambda(x,\xi) & 0 \\ 0 & \lambda(x,\xi) } \mathfrak{U}(x,\xi)^{-1},
\end{equation}
where $\mathfrak{U}$ is a unitary $2 \times 2$ matrix that depends smoothly on $(x,\xi)$. Thus, after quantization, the system $\big(\epsi D_t + \hh(x,\epsi D_x) \big) \psi = 0$ splits semiclassically near $(x_0,\xi_0)$  in two nearly decoupled equations \cite{Teufel,Martinez_Sordoni}:
\begin{equation}
   \left(\epsi D_t + \matrice{-\lambda(x,\epsi D_x) & 0 \\ 0 & \lambda(x,\epsi D_x)  } + \Or(\epsi)\right) \matrice{\phi_+ \\ \phi_-} = 0. 
   \end{equation}
According to the classical-to-quantum correspondence, the wavefront set of $\phi_t$ follows the semiclassical trajectories of $\pm \lambda(x,\xi)$ -- see e.g. \cite[Theorem 12.5]{Zw}. These form two branches $\big(x_t^+, \xi_t^+\big)$ and $\big(x_t^-, \xi_t^-\big)$, that solve respectively
\begin{equation}\label{eq:3q}
    \dfrac{dx_t^\pm}{dt} = \pm \dfrac{\p\lambda}{\p \xi}\left(x_t^\pm,\xi_t^\pm\right), \ \ \ \ \dfrac{d\xi_t^\pm}{dt} = \mp \dfrac{\p\lambda}{\p x}\left(x_t^\pm,\xi_t^\pm\right).
\end{equation}

The Hamiltonian trajectories~\eqref{eq:3q} never reach $\CC$ because (a) the energy $\pm \lambda(x_0,\xi_0) \neq 0$ is conserved along them; and (b) $\CC$ is the zero set of the function $\lambda$. 
Hence, if $(x_0, \xi_0) \notin \CC$ then the semiclassical singularities of $\phi_t$ globally evolve according to the classical-to-quantum correspondence: they follow the Hamiltonian trajectories \eqref{eq:3q} and never reach $\CC$.

Moreover, the two branches in~\eqref{eq:3q} point (at $t=0$) in opposite directions: wavepackets concentrated away from $\CC$ have no preferred direction of propagation. Their contribution to an overall quantum flux cancel out. Hence, $\CC$ is the only phase-space channel that can support unidirectional waves.

This discussion connects various characterizations of the set $\CC$:
\begin{itemize}
    \item[\textit{(i)}] \textit{Semiclassical:} $\CC$ is the set of eigenvalue crossings of $\hh(x,\xi)$;
    \item[\textit{(ii)}] \textit{Energetic:} $\CC$ is the characteristic set of $\hh(x,\xi)$, i.e. the set of points $(x,\xi)$ such that $\det \hh(x,\xi) = 0$. 
    \item[\textit{(iii)}] \textit{Topological:} the local Chern number is not defined on $\Gamma = \kappa^{-1}(0) = \pi(\CC)$ (with $\pi(x,\xi) = x$) because the eigenvalues of $\mathfrak{h}(x,\xi)$ are degenerate on $\CC$. 
    \item[\textit{(iv)}] \textit{Dynamical:} Among phase-space subsets, $\CC$ is the only (maximal) candidate that may support unidirectional wavepackets. 
\end{itemize}
Because of $(i)$, the classical-to-quantum correspondence fails. Because of conservation of energy, $(ii)$ suggests that a state semiclassically concentrated along $\CC$ should remain this way: $\CC$ acts as a semiclassical waveguide. Theorem \ref{thm:1} provides the corresponding profile and speed. Under global assumptions on $\kappa$, the bulk-edge correspondence predicts a non-vanishing quantum flux between regions of different topology. From $(iii)$, $\CC$ acts as the natural topological interface in phase-space. According to $(iv)$, it is also the only channel that can support waves contributing to a non-trivial conductivity. 

A legitimate criticism to Theorem \ref{thm:1} is that it does not study the dynamics of all initial data localized along $\CC$: it focuses on those parallel to the two-vector $[e^{-i\te_0}, - e^{i\te_0}]^\top$. As demonstrated numerically in Figure \ref{fig:6} the data prepared along the orthogonal two-vector $[-e^{i\te_0}, e^{-i\te_0}]^\top$ appear to purely disperse along the interface. An investigation of the linear case suggests that the rate of dispersion is $\epsi^{-1/4} t^{-1/2}$. 

Thus, we conjecture that general initial data semiclassically localized along $\CC$ transit to the state \eqref{eq:3s}. To write a precise statement, we split vectors $[\alpha_1,\alpha_2]^\top \in \C^2$ according to:
\begin{equation}\label{eq:0s}
    \matrice{\alpha_1 \\ \alpha_2} = \lambda_1 \matrice{e^{-i\te_0/2} \\ - e^{i\te_0/2}} + \lambda_2 \matrice{e^{-i\te_0/2} \\ e^{i\te_0/2}}.  \ \ \ 
\end{equation}
We interpret the two terms in \eqref{eq:0s} as projections on the vector from \eqref{eq:0r} and its orthogonal.

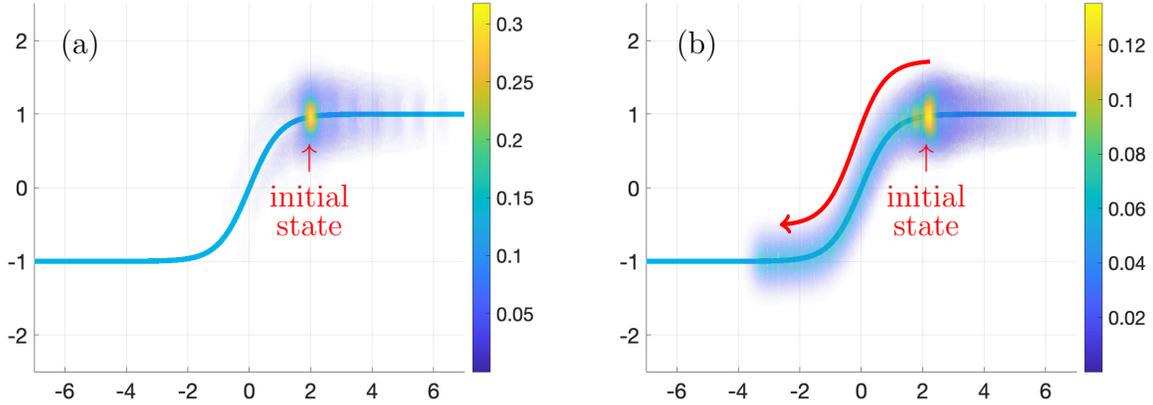
\begin{figure}[t]
{\begin{tikzpicture}
   \node at (0,0) {\includegraphics[width=8cm]{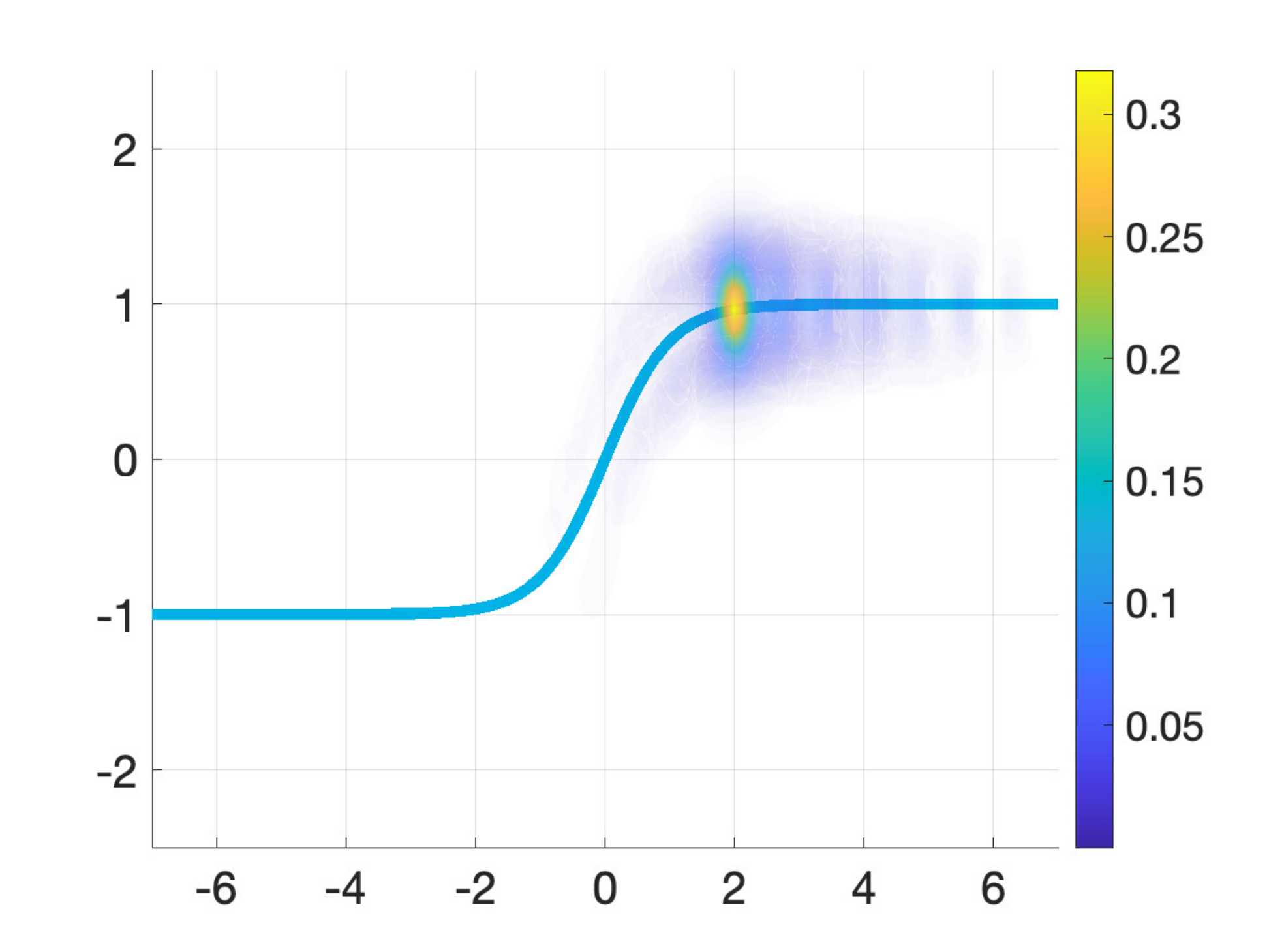}
   \includegraphics[width=8cm]{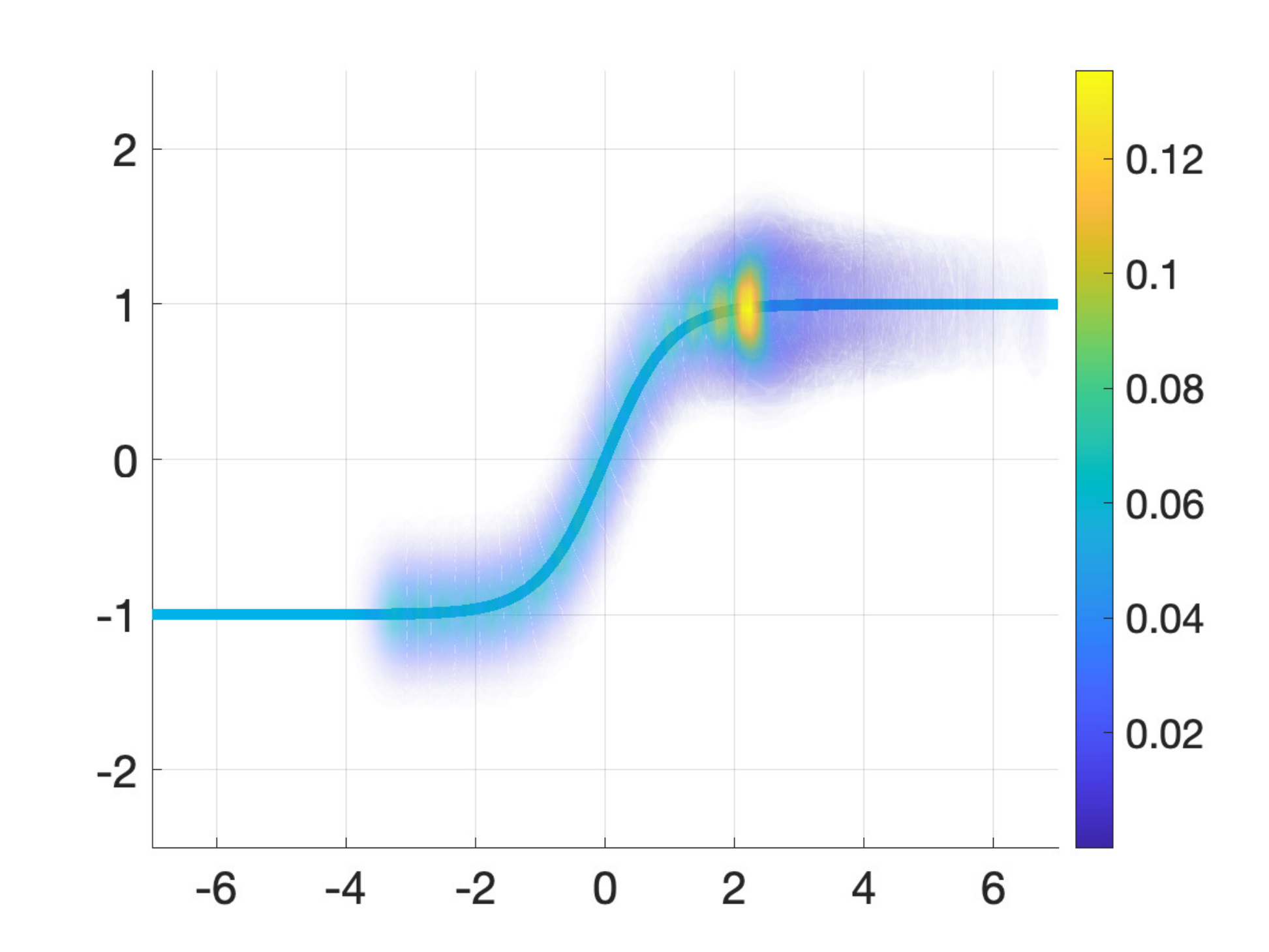}};

  \begin{scope}[shift={(4,0)}]
 \draw[domain=1:-1 , smooth, variable=\x, ultra thick,red,->] plot (\x-.2, {1.1*tanh(2.5*\x)+.7 });
 \end{scope}

 \begin{scope}[shift={(-3.45,-.2)}]
\node[red] at (0,.7) {$\uparrow$};
\node[red] at (0,.2) {initial};
\node[red] at (0,-.2) {state};
\end{scope}

\begin{scope}[shift={(4.75,-.2)}]
\node[red] at (0,.7) {$\uparrow$};
\node[red] at (0,.2) {initial};
\node[red] at (0,-.2) {state};
\end{scope}

  \node at (-6.5,2) {(a)};
  \node at (1.7,2) {(b)};
    \end{tikzpicture}}
  \caption{\label{fig:6} Solution to \eqref{eq:7n} for a $\tanh$-like interface  with (a) $[\alpha_1,\alpha_2] = [e^{-i\te_0/2},e^{i\te_0/2}]$ and (b) $[\alpha_1,\alpha_2] = [0,e^{-i\te_0/2}]$. Case (a) corresponds to $[\alpha_1,\alpha_2]$ orthogonal to the vector $[e^{-i\te_0/2},-e^{i\te_0/2}]$ from the initial data of Theorem \ref{thm:1}. This generates a purely dispersive wave along the interface. Case (b) corresponds to a linear combination of \eqref{eq:0r} and of Case (a): the solution splits into leftwards-propagating and dispersive components. 
} 
\end{figure}

\begin{conj}\label{conj:1} Fix $y_0 \in \Gamma$, $\alpha_1, \alpha_2 \in \C$, and $\lambda_1, \lambda_2$ defined according to \eqref{eq:0s}. There exists $\beta < 3/4$ such that under \eqref{eq:3r}, the solution $\Psi_t$ to 
\begin{equation}\label{eq:7n}
    (\epsi D_t + H) \Psi_t = 0, \ \ \ \ \Psi_0(x) = \dfrac{1}{{\sqrt \epsi}} \cdot
    \exp\left(-\frac{|x-y_0|^2}{2\epsi}\right) \matrice{\alpha_1 \\ \alpha_2}
\end{equation}
satisfies, uniformly in $\epsi \in (0,1]$ and $t > 0$:
\begin{equation}\label{eq:7rbis}
    \Psi_t(x) =  \dfrac{\lambda_1}{\sqrt{\epsi}} \cdot  \exp\left(-\frac{|x-y_t|^2}{2\epsi}\right) \matrice{e^{-i\te_t/2} \\ - e^{i\te_t/2}} + \Or_{L^2}\big(\epsi^{1/2}\lr{t}\big) + \Or_{L^\infty}\big(\epsi^{-\beta} \lr{t}^{-1/2} \big).
\end{equation} 
\end{conj}

The $L^\infty$-remainder in \eqref{eq:7rbis} is smaller than the leading order term as long as $\epsi^{-\beta} t^{-1/2} \ll \epsi^{-1/2}$, that is $\epsi^{1-2\beta} \ll t$. Hence, according to this conjecture, $\Psi_t$ is well approximated by the Gaussian term in \eqref{eq:7rbis} for times $\epsi^{1-2\beta} \ll t \ll \epsi^{-1/2}$ (with $\beta < 3/4$ ensuring that such times exist). This indicates that dynamical edge states generically emerge from the evolution of initial data localized along $\CC$. See \S\ref{sec:3.1} for a more general version of Conjecture \ref{conj:1}.

\subsection{Organization of the paper}

We organize the paper as follows:
\begin{itemize}
    \item In \S\ref{sec:2} we review edge state theory for Dirac operators with straight domain walls, i.e. $\kappa(x) = a \cdot x$ in \eqref{eq:1b}.
     \item  In \S\ref{sec:3} we derive the analogues of edge states for weakly curved interface. Specifically, we construct a infinite-dimensional family of solutions to $(\epsi D_t+H)\Psi_t = 0$ 
     that propagates along the topological interface $\Gamma$ for times up to $\epsi^{-1/2}$. 
     The key ingredient is a local approximation of $H$ by Dirac operators with straight interfaces. 
     \item In \S\ref{sec:4} we investigate, under a geometric condition of $\kappa$, how the curvature of $\Gamma$ affects the propagation of  wavepackets.
\end{itemize}

\subsection*{Notations} 
\begin{itemize}
\item We use $\sigma_1, \sigma_2, \sigma_3$ for the standard Pauli matrices:
\begin{equation}
    \sigma_1 = \matrice{0 & 1 \\ 1 & 0}, \ \ \ \ \sigma_2 = \matrice{0 & -i \\ i & 0}, \ \ \ \ \sigma_3 = \matrice{1 & 0 \\ 0 & -1}.
\end{equation}
\item A smooth function $f$ on $\R^2$ belongs to $C^\infty_b(\R^2)$ if it is uniformly bounded, together with its derivatives at all order.
\item A function $f \in C^\infty_b(\R^2)$ belongs to $\mathcal{S}(\R^2)$ if  $x^\alpha \p_x^\beta f$ is uniformly bounded for any $\alpha, \beta$. We provide $\SSS(\R^2)$ with the family of seminorms $\big|x^\alpha \p_x^\beta f\big|_{L^\infty}$.
\item The operators $D_{x_j}$ and $D_t$ are defined by  $D_{x_j} = -i\partial_{x_j}$ and $D_t = -i \partial_t$.
\item We use the japanese bracket notation: $\lr{x} = \sqrt{1+|x|^2}$. 
\item We denote by $\ker_\VV(A)$ the kernel of a linear operator $A$ acting on a vector space $\VV$. 
\item If $v \in \R^2$, $v^\perp$ is the counterclockwise $\pi/2$-rotation of $v$. 
\item $\lr{u,v}_{L^2} = \int_{\R^2} \overline{u} v.$ 
\item For $f$ in a normed vector space $\mathcal{X}$, we write $f=\Or_{\mathcal{X}}(\epsi)$ if $|f|_{\mathcal{X}} \leq C \epsi$ for some constant $C>0$ independent of $\epsi$.
\item Given $\alpha \in \C^2$, $\alpha^\perp = -i \sigma_2 \alpha$ is the $\pi/2$-rotation of $\alpha$. 
\item $y_t$ is the solution to the ODE \eqref{eq:3r} with initial data $y_0 \in \Gamma$; $\theta_t$ is the angle between the $y$-axis and $\nabla \kappa(y_t)$; and $r_t = \big|\nabla \kappa(y_t)\big|$. See Figure \ref{fig:4}.
\end{itemize}

  \subsection*{Acknowledgments} This work started during the AIM workshop Mathematics of topological insulators. The authors gracefully thank the organizers: Daniel Freed, Gian Michele Graf, Rafe Mazzeo and Michael Weinstein. They also thank  Mitchell Luskin and Cl\'ement Tauber for interesting discussions. The authors acknowledge support form the NSF grants DMREF-1922165 (AW), DMS-2118608 (AD), DMS-1908736 (GB), EFMA-1641100 (GB) and DMS-2012286 (JL); the EPSRC grant EP/L016516/1 (SB); the U.S. Department  of  Energy  grant  DE-SC0019449 (JL); the Office of Naval Research grant N00014-17-1-2096 (GB); and the ARO MURI grant W911NF-14-0247 (AW).

\section{Edge states and dynamics for straight interfaces}\label{sec:2} 

We review here the simplest example of domain wall $\kappa$: we write
\begin{equation}
    \kappa(x) = \kappa_{\te,r}(x) = -r\sin(\te) x_1 + r\cos(\te) x_2 = r   \matrice{- \sin(\te) \\ \cos(\te) } \cdot  x
\end{equation}
with $\te \in \R$, $r > 0$. The interface $\kappa_{\te,r}^{-1}(0)=\R v_{\theta}$ is a straight line, directed by the vector $v_{\theta} = -[\cos(\te),\sin(\te)]^\top$ -- see Figure \ref{fig:3}. The Hamiltonian is then
\begin{equation}\label{eq:Hte}
    H_{\te,r} =   \matrice{
\kappa_{\te,r}(x) & \epsi D_{x_1} - i\epsi D_{x_2} 
\\
\epsi D_{x_1} +i \epsi D_{x_2} & - \kappa_{\te,r}(x) }.
\end{equation}
It admits edge states: solutions to  $(H_{\te,r} - \lambda) F_{\te,r} {\color{blue} = 0}$ that are localized and harmonic along $\R v_\te$. Here we review their explicit expression and their dynamical properties.

\begin{figure}[b]
\floatbox[{\capbeside\thisfloatsetup{capbesideposition={right,center},capbesidewidth=3in}}]{figure}[\FBwidth]
{\hspace{-1cm}\caption{
Currents propagate along $\Gamma$ at speed $v_{\theta}$ given by the counterclockwise rotation of $\nabla \kappa$.
}\label{fig:3}}
{\begin{tikzpicture}
   \node at (0,0) {\includegraphics[]{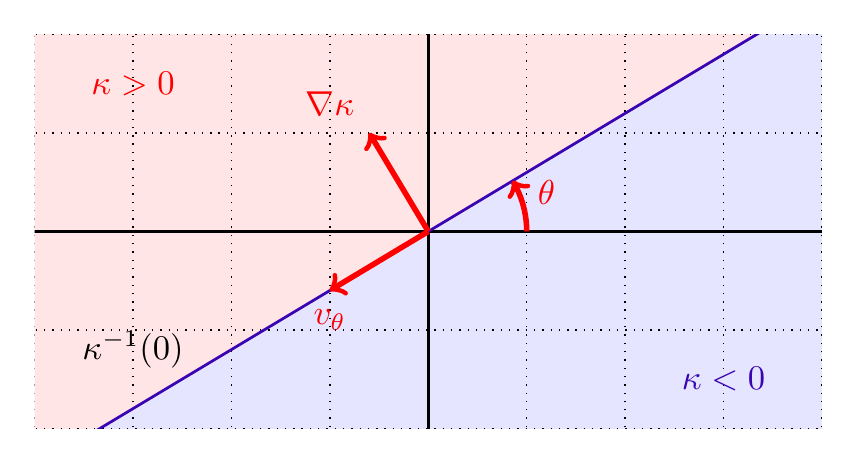}};
  \end{tikzpicture}}
\end{figure}

\subsection{Conjugation properties} We first show that the Hamiltonians $H_{\theta,r}$ and $H_{0,r}$ are conjugated by a change of frame and gauge. For this purpose, we introduce the operator
\begin{equation}\label{eq:3b}
   \UU_\te f(x) = U_\te f(R_\te x), \ \ \ \      R_\te =    \matrice{ \cos(\te) & \sin(\te)  \\ -\sin(\te) & \cos(\te)}; \ \ \ \ 
      U_\te = \matrice{ e^{-i\te/2} & 0 \\ 0 & e^{i\te/2} }.
\end{equation}

\begin{lemm} \label{lem:equivalence} The Hamiltonian \eqref{eq:Hte} is unitarily equivalent to the Hamiltonian $H_{0,r}$ with 
\begin{equation}
 \UU_\te^{-1} H_{\te,r} \, \UU_\te = H_{0,r}.
\end{equation}
\end{lemm}

\begin{proof} Let $\RR_\te$ be the pullback operator by $R_\te$: $\RR_\te f(x) = f(R_\te x)$. We note that $\kappa_{\te,r}(x) = r \cdot R_\te^\top e_2 \cdot x = r (R_\te x)_2$. Thus $\RR_\te^{-1} \kappa_{\te,r} \RR_\te = r x_2$. 
 We now use $\RR_\te^{-1} D_x \RR_\te = R_\te^\top D_x$ to compute partial derivatives involved in $H_{\te,r}$:
 \begin{equation}
 \begin{split}
    \RR_\te^{-1}( D_{x_1} +i D_{x_2}) \RR_\te= \matrice{1 \\ i} \cdot R_\te^\top D_x &= R_\te \matrice{1 \\ i} \cdot D_x = \matrice{e^{i\te} \\ i e^{i\te} } \cdot D_x = e^{i\te} ( D_{x_1} +i D_{x_2}).
    \end{split}
 \end{equation}
The adjoint identity is
\begin{equation}
    \RR_\te^{-1}( D_{x_1}-i D_{x_2}) \RR_\te = e^{-i\te} ( D_{x_1} -i D_{x_2}).
\end{equation} 
 Grouping these identities, we obtain:
 \begin{equation}
 \begin{split}\label{eq:1i}
    \RR_\te^{-1} H_{\te,r} \RR_\te 
    &= \matrice{ rx_2 & e^{-i\te} \epsi (D_{x_1} - iD_{x_2})  \\
e^{i\te} \epsi (D_{x_1} +i D_{x_2}) & - rx_2 } \\
&= s_1 \epsi D_{x_1} + s_2 \epsi D_{x_2} + s_3 r x_2,
\end{split}
 \end{equation}
 where, $s_1, s_2, s_3$ are $2 \times 2$ Hermitian matrices given by
\begin{equation}
\label{eq:derPauli}
    s_1 = \left[  \begin{matrix} 0 & e^{-i\te} \\ e^{i \te} & 0
 \end{matrix} \right], \ \ \ \ s_2 = \left[  \begin{matrix} 0 & -i e^{-i\te} \\ i e^{i \te} & 0  \end{matrix} \right], \ \ \ \ s_3 = \left[  \begin{matrix} 1 & 0  \\ 0 & -1  \end{matrix} \right] = \sigma_3.
\end{equation}
An explicit calculation shows that $U_\te^{-1} s_j  U_\te= \sigma_j$. We conclude that
\begin{equation}
\label{eq:equiva}
    \UU_\te^{-1} H_{\te,r} \UU_\te = \sigma_1 \epsi D_{x_1} + \sigma_2 \epsi D_{x_2} + \sigma_3 x_2 = H_{0,r}.
\end{equation}
 This completes the proof. \end{proof}

\begin{rem}\label{rem:1} The relation \eqref{eq:3b} allows us to calculate the conductivity of $H_{\te,r}$ in the direction of $v_\te$, see \eqref{eq:0l}: it is equal to $1$. Indeed, the conductivity of $H_{0,r}$ (counted positively in the direction of $e_2^\perp = -e_1$) is equal to $1$ \cite{B19b}.
 Therefore, using invariance of the trace under conjugation, and the fact that $f$ is a scalar function:
\begin{align}
    1 & = \operatorname{Tr}_{L^2} \Big( \big[H_{0,r}, f(-x_1)\big] g'(H_{0,r}) \Big)
    \\
\label{eq:0o}    & = \operatorname{Tr}_{L^2} \Big( \big[H_{\te,r}, \RR_\te f(-x \cdot e_1) \RR_\te^{-1}\big] g'(H_{\te,r}) \Big) = \operatorname{Tr}_{L^2} \Big( [H_{\te,r}, f(v_\te \cdot x) ] g'(H_{\te,r}) \Big).
\end{align}
\end{rem}
 
The Hamiltonian $H_{0,r}$ admits edge states: for any $\xi \in \R$, if  
\begin{equation}
    F_{0,r}(\xi,x) = \exp\left( \frac{i\xi x_1}{\epsi} -  \dfrac{r  x_2^2}{2\epsi} \right) \matrice{1 \\ -1}, 
\end{equation}
then $F_{0,r}(\xi,\cdot)$ is a plane wave in $x_1$, i.e. along the interface; 
decays transversely along the interface, i.e., in $x_2$; and satisfies the stationary Dirac equation $(H_{0,r} - \xi) F_{0,r}(\xi, \cdot) = 0$. From Lemma \ref{lem:equivalence} we deduce that $H_{\te,r}$ also admits edge states: 
\begin{equation}\label{eq:7k}
     F_{\te,r}(\xi, x) = \UU_\te F_{0,r}(\xi,x) =  \exp\left( \frac{i\xi (R_\te x)_1}{\epsi} - \dfrac{r (R_\te x)_2^2}{2 \epsi} \right) \matrice{
e^{-i\te/2} \\ -e^{i\te/2} }.
\end{equation}

\subsection{Dynamics of edge states.} We review here how edge states give rise to an infinite-dimensional family of ballistic waves for Dirac operators with linear domain walls.

\begin{prop}
\label{prop:mainresultlinear} For any $f \in \SSS(\R)$, the function 
\begin{equation}\label{eq:1c}
     \psi_t(x)=  \epsi^{-1/2} \cdot f \big( t+ (R_\te x)_1 \big) \cdot \exp\left( - \dfrac{r (R_\te x)_2^2}{2\epsi} \right) \matrice{e^{-i\te/2} \\ -e^{i\te/2}}
\end{equation}
solves the equation $(\epsi D_t + H_{\te,r}) \psi_t = 0$.
\end{prop}

The functions \eqref{eq:1c} are the ballistic waves generated by edge states: they  propagate along the interface $\R v_\te$ and decay rapidly along $\R v_\te^\perp$. 
 Our scaling casts \eqref{eq:1c} as 
wavepackets:
\begin{equation}\label{eq:7l}
    \psi_t(x) =  \epsi^{-1/2} \cdot a \left(\dfrac{x - y_t}{\sqrt{\epsi}} \right), \ \ \ \ a(y) = e^{-\frac{r}{2} (R_\te y)_2^2 }f \big( \sqrt{\epsi} (R_\te y)_1 \big)\matrice{e^{-i\te/2} \\ -e^{i\te/2}}, \ \ \ \ y_t = t v_\te
\end{equation}
with $a$ having a full asymptotic expansion in powers of $\sqrt\epsi$.
This connection will be the basis of our analysis in the context of curved interfaces.

\begin{proof}[Proof of Proposition \ref{prop:mainresultlinear}] Let  $g \in \SSS(\R)$ such that 
\begin{equation}
    g(\xi) = \frac{1}{2\pi\epsi} \int_\R e^{\frac{i}{\epsi} t \xi} f(t) dt.
\end{equation}
We introduce 
\begin{equation}\label{eq:7f}
\psi_t(x)=\epsi^{-1/2} \int_\R e^{-\frac{i}{\epsi} t\xi} g(\xi) F_{\te,r}(\xi,x) d\xi.
\end{equation}
Since $(H_{\te,r} - \xi) F_{\te,r}(\xi,\cdot) = 0$, 
we deduce that
\begin{equation}
\begin{aligned}
\epsi D_t \psi_t(x) & =  - \epsi^{-1/2} \int_\R \xi e^{-\frac{i}{\epsi} t\xi} g(\xi) F_{\te,r}(\xi,x) d\xi \\
& = 
- \epsi^{-1/2} \int_\R e^{-\frac{i}{\epsi}  t\xi} g(\xi)H_{\te,r} F_{\te,r}(\xi,x) d\xi = - H_{\te,r}\psi_t(x).
\end{aligned}
\end{equation}
This proves that \eqref{eq:7f} is a solution to $(\epsi D_t + H_{\te,r}) \psi_t = 0$. Plugging the formula \eqref{eq:7k} for $F_{\te,r}$ in \eqref{eq:7f}, we obtain
\begin{align}
\psi_t(x) & = \epsi^{-1/2} \int_\R e^{- \frac{i}{\epsi} (  t+(R_\te x)_1)\xi} g(\xi)  d\xi \cdot \exp\left( - \dfrac{r (R_\te x)_2^2}{2\epsi} \right) \matrice{
e^{-i\te/2} \\ -e^{i\te/2} }\\
&=  \epsi^{-1/2} f\big(t+(R_\te x)_1 \big) \exp\left( - \dfrac{r (R_\te x)_2^2}{2\epsi} \right) \matrice{
e^{-i\te/2} \\ -e^{i\te/2} },
\end{align}
by definition of $g$ as the inverse (semiclassical) Fourier transform of $f$. 
\end{proof}

\section{Dynamical analogues of edge states along curved interfaces}\label{sec:3} We now consider non-linear domain walls, opening the possibility for curved topological interfaces. We relax \eqref{eq:3r} to a global transversality condition: 
\begin{equation}\label{eq:7z}
\inf \big\{ \big| \nabla \kappa(y) \big| : \ \kappa(y) = 0\} > 0.
\end{equation}
We recall that all derivatives of $\kappa$ are uniformly bounded: $\kappa \in C^\infty_b(\R^2)$. 
We plan to produce a dynamical analogue of edge states: a solution to
\begin{equation}\label{eq:Dirac}
(\epsi D_t + H) \psi = 0, \ \ \ \ H = \matrice{ \kappa(x) & \epsi D_{x_1} - i\epsi D_{x_2}  \\ \epsi D_{x_1} + i\epsi D_{x_2} & -  \kappa(x) },
\end{equation}
that propagates for long time along the topological interface $\Gamma = \kappa^{-1}(0)$.

The equation \eqref{eq:7l} motivates the ansatz
\begin{equation}\label{eq:1d}
    \psi(t,x) =  \epsi^{-1/2} a\left(t, \dfrac{x - y_t}{\sqrt{\epsi}} \right), \ \ \ \ \text{where:}
\end{equation}
\begin{itemize}
\item $a \in \SSS(\R^2,\C^2)$ has a full expansion in powers of $\epsi^{1/2}$;
\item $y_0 \in \Gamma$ and $y_t \in \Gamma$ is the solution of the ODE
\begin{equation}
    \dot{y_t} = v(y_t), \ \ \ \ v(y) = \dfrac{\nabla \kappa(y)^\perp}{|\nabla \kappa(y)|}, \ \ \ \ w^\perp = \matrice{0 & -1 \\ 1 & 0} w. 
\end{equation}
\end{itemize}
The vector $v(y)$ is the local analogue to $v_\te$: at each point $y \in \Gamma$, it is the unit tangent vector to $\Gamma$ obtained by rotating counterclockwise $\nabla \kappa(y)$. Since $\kappa(y_0) = 0$, 
  $y_t \in \Gamma$ for any $t$:
\begin{equation}
    \frac{d \kappa(y_t)}{dt}  = \dot{y_t} \cdot \nabla\kappa(y_t) = v(y_t)  \cdot \nabla\kappa(y_t) = 0.
\end{equation}
Let $\te_t$ and $r_t$ be such that 
\begin{equation}\label{def:theta_t}
\nabla \kappa (y_t) = r_t \matrice{-\sin(\te_t) \\ \cos(\te_t)}, \ \ \ \ \text{so that}  \ \ v(y_t)=-\matrice{\cos(\te_t) \\ \sin(\te_t)},
\end{equation}
see Figure \ref{fig:4}. With these notations in place, we define $\KK_t : \SSS(\R) \rightarrow \SSS(\R^2,\C^2)$ by:
\begin{equation}\label{def:a_0}
\KK_t f(x) = r_t^{1/4} f \big((R_{\theta_t}x)_1\big) e^{-\frac{r_t} 2 (R_{\theta_t} x)^2_2}  \matrice{e^{-i\te_t/2} \\ -e^{i\te_t/2} }, \ \ \ \ f \in \SSS(\R). 
\end{equation}

\begin{theo}\label{thm:2} Let $\kappa \in C^\infty_b(\R^2)$ satisfying \eqref{eq:7z} and $y_t$, $\theta_t$ as above. Let $\psi_t$ be the solution to $(\epsi D_t + H) \psi_t = 0$ with 
\begin{equation} \label{eq:initial_data}
    \psi_0(x) = \dfrac{1}{\sqrt{\epsi}} \cdot \KK_0 f \left(\dfrac{ x - y_0}{\sqrt{\epsi}} \right), \ \ \ \ f \in \SSS(\R).
\end{equation}
Then uniformly for $\epsi \in (0,1]$ and $t > 0$:
\begin{equation} \label{eq:psi_with_error}
    \psi_t(x) = \dfrac{1}{\sqrt{\epsi}} \cdot \KK_t f \left(\dfrac{x - y_t}{\sqrt{\epsi}} \right) + \Or_{L^2}\big(\epsi^{1/2} \lr{t}\big).
\end{equation}
\end{theo}

Theorem \ref{thm:2} constructs a solution to $(\epsi D_t + H) \psi_t = 0$, propagating dispersion-free along $y_t$, for times $t \ll \epsi^{-1/2}$. Under geometric conditions on $\kappa$, we can extend this time of validity; see Theorem \ref{thm:3}. These two results focus on maximizing the lifespan of approximate solutions. We can instead focus on improving their accuracy: see Theorem \ref{thm:4} for solutions up to $O(\epsi^n)$ for every $n$, but fixed lifetime. 
 
When $r_t$ is not constant -- corresponding to \eqref{eq:7z} holding instead of \eqref{eq:3r} -- the state in \eqref{eq:psi_with_error} is coherent in a relaxed sense: there may be lateral spreading at scale $r_t$ (which remains bounded above and below by our assumptions on $\kappa$). See the expression \eqref{def:a_0} for $\KK_t f$ and Figure \ref{fig:1} for a numerical illustration.

\begin{figure}[t]
\floatbox[{\capbeside\thisfloatsetup{capbesideposition={right,center},capbesidewidth=3.2in}}]{figure}[\FBwidth]
{\hspace{-1cm}\caption{\label{fig:1}
A straight interface but a non-linear domain wall: $\kappa(x) = (1-0.9 \sin(x_1))x_2$. We have $y_t = -t e_1$ hence $r_t = 1 + 0.9 \sin(t)$. This quantity nearly degenerates for $t$ near $-\pi/2+\pi\mathbb{Z}$, inducing lateral spreading of the wavepacket for such times, but reconstruction in between.
}}
{\begin{tikzpicture}
   \node at (0,0) {\includegraphics[width=10cm]{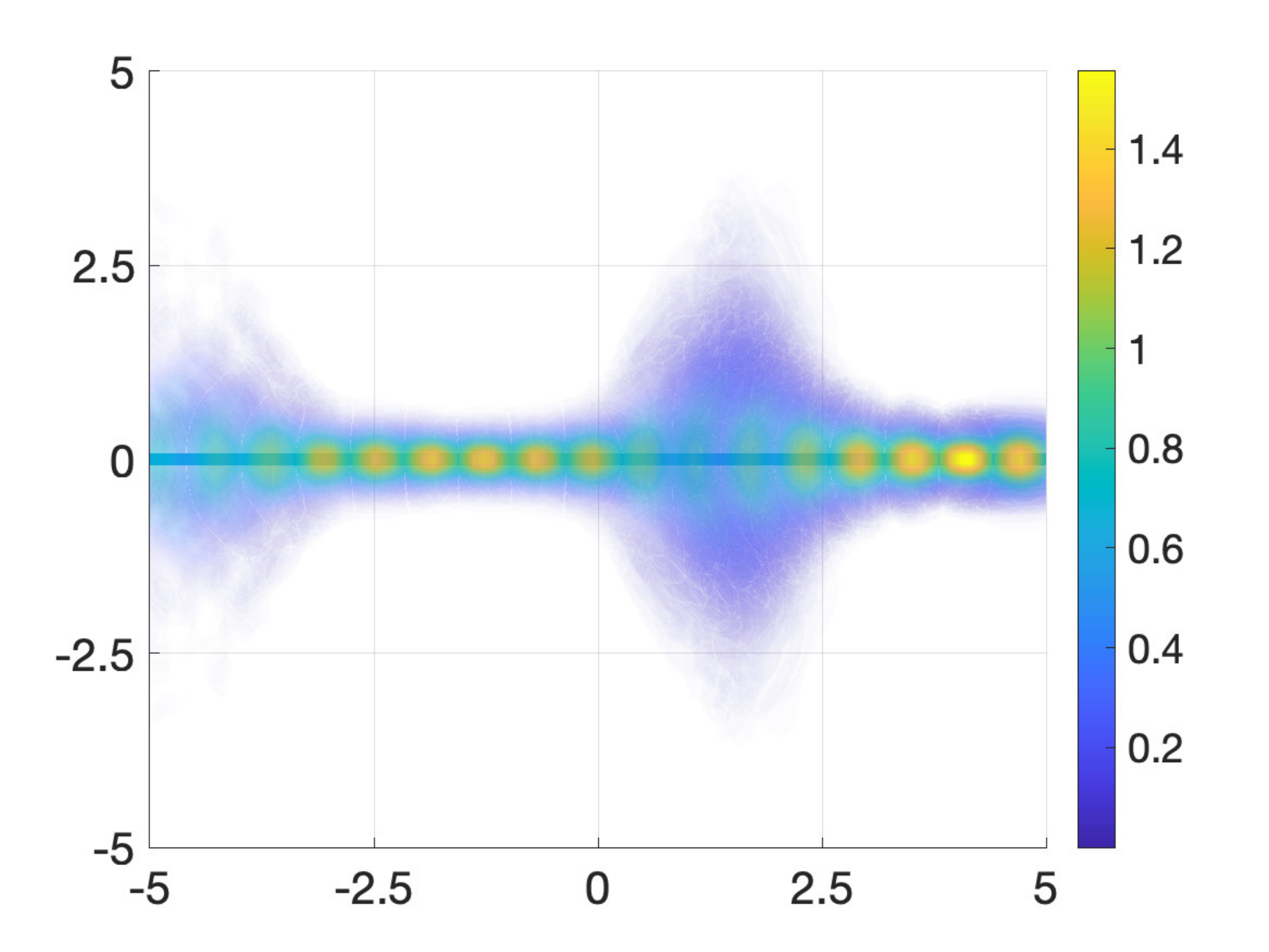}};
  \end{tikzpicture}}
\end{figure}

The initial data \eqref{eq:initial_data} is quite specific: the rescaled amplitude $\KK_0 f$ is in the range of $\KK_0$. To obtain a full picture of evolution of states initially microlocalized along $\CC$, we need to understand how orthogonal initial data propagate:
\begin{equation}
    \psi_0(x) = \dfrac{1}{\sqrt{\epsi}} \cdot \KK_0 f \left(\dfrac{ x - y_0}{\sqrt{\epsi}} \right)^\perp.
\end{equation}
This suggests a refinement of Conjecture \ref{conj:1}. Let $\Pi : \SSS(\R^2,\C^2) \rightarrow \SSS(\R^2,\C^2)$ be the orthogonal projection on the range of $\KK_0$. We observe that $\KK_0$ is an isomorphism to its range; therefore, for any $a \in \SSS(\R^2,\C^2)$, there exists a unique $f \in \SSS(\R)$ such that $\Pi a = \KK_0 f$.

\begin{conj}\label{conj:2} There exists $\beta < 3/4$ with the following. Let $a \in \SSS(\R^2,\C^2)$, $f \in \SSS(\R)$ such that $\Pi a = \KK_0 f$, and $\phi_t$ be the solution to $(\epsi D_t + H) \phi_t = 0$ with initial data
\begin{equation}
    \phi_0(x) = \dfrac{1}{\sqrt{\epsi}} \cdot a\left( \dfrac{x - y_0}{\sqrt{\epsi}} \right).
\end{equation}
Then uniformly in $\epsi \in (0,1]$, $t > 0$:
\begin{equation}\label{eq:0u}
  \phi_t(x) =  \dfrac{1}{\sqrt{\epsi}} \cdot  \KK_t f \left(\dfrac{x - y_t}{\sqrt{\epsi}} \right) + \Or_{L^2}\left(\epsi^{1/2} t\right)
 + \Or_{L^\infty}\left(\epsi^{-\beta} t^{-1/2}\right).  
 \end{equation}
\end{conj}

According to Conjecture \ref{conj:2},  any function localized (in a semiclassical sense) near $(y_0,0)$ splits in propagating and dispersive parts, with the analogue of an edge state emerging dynamically. See Figure \ref{fig:6} for a numerical confirmation.

\subsection{Structure of proof of Theorem \ref{thm:2}}

We will prove  Theorem \ref{thm:2} by establishing the following statements.
\begin{enumerate}
\item Approximate solutions of the Dirac equation solve a hierarchy of transport equations, see Lemma \ref{lem:1e}.
\item The leading-order transport operator has explicit kernel and a spectral gap away from its kernel, see \S\ref{sec:3.1}.
\item Solutions to the hierarchy of transport equations exist, see \S\ref{sec:3.2}-\S\ref{sec:5.3}.
\item Approximate and exact solutions to the Dirac equation are nearly equal, see \S\ref{sec:3.6}.
\end{enumerate}
We will use the notation  
\begin{equation}\label{def:W[]}
    W[a]_{y_t}(x) = \frac{1}{\sqrt{\epsi}} \cdot a\left( \dfrac{x - y_t}{\sqrt{\epsi}} \right)
\end{equation}
for $a \in \mathcal{S}(\R^2,\C^2)$ possibly depending on $t$ and $\epsi$.

We also introduce the operators $T_j$ acting on $\SSS(\R^2,\C^2)$, defined by:
\begin{align}\label{def:T0}
 T_0 & =  - \dot{y_t} \cdot D_x  +\left[  \begin{matrix}
 \nabla \kappa(y_t) x & D_{x_1} - iD_{x_2}
\\ D_{x_1} + iD_{x_2} & -\nabla \kappa(y_t) x
 \end{matrix} \right],\\
 \label{def:T1}
 T_1 & =D_t +  \left( \sum_{|\alpha|=2} \frac 1 {\alpha!} \partial^\alpha \kappa(y_t) x^\alpha \right) \sigma_3, \\
 \label{def:Tj}
 T_j& = \left( \sum_{|\alpha|=j+1} \frac 1 {\alpha!} \partial^\alpha \kappa(y_t) x^\alpha \right) \sigma_3,\;\;j\geq 2. 
\end{align}

\subsection{Formal approximate solutions via transport equations} \label{sec:hierarchy}

We start with the following lemma: solving the hierarchy of transport equations 
\begin{equation}\label{eq:3g}
  T_0 a_0 = 0, \ \ \ \ T_0 a_1 + T_1 a_0 = 0, \ \ \ \ \dots, \ \ \ \  \sum_{\ell=0}^j T_{j-\ell}  a_{\ell} = 0, \ \ \ j \in [0,m]
\end{equation}
produces approximate solutions to the Dirac equation.

\begin{lemm}\label{lem:1e} For any $m \in \N$, there exists $C > 0$ such that if $a_0$, $a_1$, \dots, $a_m \in \SSS(\R^2,\C^2)$ are solutions of \eqref{eq:3g} and $a^{(m)} = \sum_{\ell=0}^m \epsi^{\ell/2} a_\ell$, then for all $\epsi \in (0,1]$:
\begin{equation} \label{eq:2j}
   \left\| (\epsi D_t + H) W\big[a^{(m)}\big]_{y_t}  \right\|_{L^2}  \leq C \epsi^\frac{m+2}{2} \left( \| D_t a_m \|_{L^2} + \sum_{k = 0}^m  \big\| \lr{x}^{m+2} a_k \big\|_{L^2} \right).
\end{equation}
\end{lemm}

\begin{proof}[Proof of Lemma \ref{lem:1e}] \textbf{1.} Fix $m \in \N$. We observe that for $a\in\mathcal S(\R^2,\C^2)$,
\begin{equation}\label{estimation1}
    \epsi \p_{x_j} W[a]_{y_t} = W[\sqrt{\epsi}\, \p_{x_j} a]_{y_t}, \;\;\epsi D_t W[a]_{y_t} = W[-\sqrt{\epsi} \,\dot{y_t} \cdot D_x a + \epsi D_t a]_{y_t}.
\end{equation}
We now write the Taylor--Lagrange identity with integral remainder (note that $\kappa(y_t) = 0$): 
\begin{align}
    \kappa(x) &    = \left( \sum_{|\alpha| = 1}^{m+1} \dfrac{1}{\alpha !} \p^\alpha \kappa(y_t) (x - y_t)^\alpha \right)  + r_m(x-y_t), \ \ \ \ \text{with} \\
   r_m(x) & = \dfrac{1}{(m+1)!} \sum_{|\alpha| = m+2} x^\alpha \int_0^1 (1-s)^{m+1} \p^\alpha \kappa\big( y_t + sx \big) ds.
   \end{align}
We deduce that
\begin{align}\label{estimation2}
    \kappa(x) W[a]_{y_t}(x)
    &= W\left[  \left( \sum_{|\alpha| = 1}^{m+1} \dfrac{\epsi^{|\alpha|/2}}{\alpha !} \p^\alpha \kappa(y_t) x^\alpha  + \epsi^{\frac{m+2}{2}} R_m(x) \right) a \right]_{y_t}(x), \ \ \ \ \text{with} \\
    R_m(x) & = \epsi^{-\frac{m+2}{2}} r_m(\epsi^{1/2} x) = \dfrac{1}{(m+1)!} \sum_{|\alpha| = m+2} x^\alpha \int_0^1 (1-s)^{m+1} \p^\alpha \kappa\big( y_t + s \epsi^{1/2} x \big) ds.
    \end{align}
    Since $\big|r_m(x)\big| \leq C |x|^m$, we obtain that $R_m(x) \leq C|x|^m$ for all $\epsi \in (0,1]$. From the relations \eqref{estimation1}-\eqref{estimation2} and the definition \eqref{def:Tj} of the operators $T_j$:
\begin{equation}
    (\epsi D_t + H) W[a]_{y_t} = W\left[ \left(\sum_{j = 0}^{m} \epsi^{\frac{j+1}{2}} T_j + \epsi^{\frac{m+2}{2}} R_m \right)a \right]_{y_t}. 
\end{equation}
In particular, using that $W[a]_{y_t}$ and $a$ have the same $L^2$-norm,
\begin{equation}\label{eq:2n}
    \big\| (\epsi D_t + H) W[a]_{y_t} \big\|_{L^2} = \left\| \left(\sum_{j = 0}^{m} \epsi^{\frac{j+1}{2}} T_j + \epsi^{\frac{m+2}{2}} R_m \right) a \right\|_{L^2}.
\end{equation}

\textbf{2.} Assume now that $a_j$ solves the equations \eqref{eq:3g}, and plug $a^{(m)} = \sum_{k=0}^m \epsi^{k/2} a_k$ for the amplitude in \eqref{eq:2n}. Then we obtain:
\begin{align}
    \big\| (\epsi D_t + H) W[a^{(m)}]_{y_t} \big\|_{L^2}
    & 
    = \left\| \sum_{\substack{j,k = 0\\ j+k \geq m+1} }^m \epsi^{\frac{j+k+1}{2}} T_j a_k + \sum_{k=0}^m \epsi^{\frac{m+2+k}{2}} R_m a_k \right\|_{L^2} 
    \\
    & 
    \leq \sum_{\substack{j,k = 1\\ j+k \geq m+1} }^m \epsi^{\frac{j+k+1}{2}} \big\| T_j a_k \big\|_{L^2} + \sum_{k=0}^m \epsi^{\frac{m+2+k}{2}} \big\| R_m a_k \big\|_{L^2}.
\end{align}
In the second line we used the first sum starts at $j, k = 1$, since $j+k \geq m+1$ and $j, k \leq m$. 

We note that $T_1$ is the sum of $D_t$ and a polynomial of degree $2$. For $j \geq 2$, $T_j$ is a  polynomial of degree $j+1$; and $R_m$ is bounded by $C|x|^{m+2}$. All coefficients involved depend on derivatives of $\kappa$; in particular their values at $y_t$ are uniformly bounded in time. In particular, after extracting $D_t$, we can bound all multiplicative terms by $C\lr{x}^{m+2}$. We obtain that $ \big\| (\epsi D_t + H) W[a^{(m)}]_{y_t} \big\|_{L^2}$ is bounded, up to a multiplicative constant, by:
\begin{equation}
    \epsi^\frac{m+2}{2} \| D_t a_m \|_{L^2} + \sum_{\substack{j,k = 1\\ j+k \geq m+1} }^m \epsi^{\frac{j+k+1}{2}} \big\| \lr{x}^{m+2} a_k \big\|_{L^2} + \sum_{k=0}^m \epsi^{\frac{m+2+2k}{2}} \big\| \lr{x}^{m+2} a_k \big\|_{L^2}.
 \end{equation}
 Noting that $j+k + 1 \geq m+2$ in the first sum, we conclude that for any $t$,
\begin{equation}
   \big\| (\epsi D_t + H) W[a^{(m)}]_{y_t} \big\|_{L^2} \leq C \epsi^\frac{m+2}{2} \left( \| D_t a_m \|_{L^2} + \sum_{k = 0}^m  \big\| \lr{x}^{m+2} a_k \big\|_{L^2} \right).
\end{equation}
This completes the proof. \end{proof}

We will show in the following how to construct solutions $a_j$ to the hierarchy \eqref{eq:3g}, and then bound their derivatives and moments. Together with Lemma \ref{lem:1e} this will give a rigorous construction of approximate solutions to the Dirac equation.
 
\subsection{Spectral analysis of leading order transport operator} \label{sec:3.1} 
The dominant equation of the hierarchy \eqref{eq:3g} is $T_0 a_0=0$, where $T_0$ is defined in \eqref{def:T0}; the other equations are
$$T_0 a_j = -\sum_{\ell=0}^{j-1}  T_{j-\ell } a_\ell, \quad 1 \leq j \leq m. $$

Solving these equations amounts to (i) find $\ker(T_0)$; and (ii) establish a stability estimate (here, a spectral gap) for $T_0^{-1}$ away from $\ker(T_0)$. Below we write $T_0 = L_{\te_t,r_t}$, where
\begin{equation} \label{eq:Lte}
    L_{\te,r} = \matrice{ \cos(\te) \\  \sin(\te) }D_x + \matrice{
 r\kappa_{\te,r}(x) & D_{x_1} - iD_{x_2}
\\ D_{x_1} + iD_{x_2} & -r\kappa_{\te,r}(x)} = \matrice{ \cos(\te) \\  \sin(\te) }D_x + H_{r,\theta}. 
\end{equation}
We now focus on the analysis  of $L_{\theta,r}$ on $\mathcal S(\R^2,\C^2)$. We first compute its kernel (Lemma~\ref{prop:1}) and  prove it is one to one on the orthogonal complement (Lemma~\ref{lem:1b}). 

\begin{lemm}\label{prop:1} For every $r > 0$ and $\te \in \R$, the nullspace of $L_{\theta,r} : \SSS(\R^2,\C^2) \rightarrow \SSS(\R^2,\C^2)$ is  \begin{equation}  \label{eq:1j}
     \ker_{\SSS(\R^2)} (L_{\te,r}) = \left\{ 
     f\big((R_\te x)_1\big) e^{-\frac{r(R_\te x)_2^2}{2}} \matrice{e^{-i\te/2} \\ -e^{i\te/2} }, \ f\in \SSS(\R)
     \right\}.
 \end{equation}
\end{lemm}

\begin{proof} As in \eqref{eq:equiva}, $\UU_\te^{-1} L_{\te,r} \UU_\te = L_{0,r}$, with $\UU_\te = \RR_\te U_\te$. Indeed, $U_\te^{-1}  \RR_\te^{-1} H_\te \RR_\te U_\te = H_0$ and 
\begin{equation}
    \UU_\te^{-1}  \matrice{ \cos(\te) \\ \sin(\te)} \cdot D_x \UU_\te = R_\te^\top  e_1 \cdot R_\te^\top D_x = D_{x_1}. 
\end{equation}
Moreover, if $S_r f(x) = f(\sqrt{r}x)$, then we have 
\begin{equation}\label{eq:3p}
    S_r^{-1} H_{0,r} S_r = \sqrt{r} H_{0,1}.
\end{equation}
Hence, $H_{\te,r}$ and $H_{0,1}$ are conjugated (up to multiplication by $\sqrt{r}$). The identity \eqref{eq:Lte} {implies}
%imply 
that the same holds for $L_{\te,r}$ and $L_{0,1}$:
\begin{equation}\label{eq:2k}
    S_r^{-1} \UU_\te^{-1} L_{\te,r} \UU_\te S_r = \sqrt{r} L_{0,1}.
\end{equation}

Thus, to find the kernel of $L_{\te,r}$, it suffices to find that of $L_{0,1}$. We have 
\begin{equation}
    L_{0,1} = \matrice{D_{x_1} + x_2 & D_{x_1} - iD_{x_2} \\ D_{x_1} + iD_{x_2} &  D_{x_1} - x_2} = \matrice{1 & 1 \\ 1 &  1} D_{x_1} + \matrice{ x_2 &  - iD_{x_2} \\  iD_{x_2} &   - x_2}.
\end{equation}
We claim that
\begin{equation}\label{eq:1m}
    \ker_{\SSS(\R^2)}(L_{0,1}) = \left\{ f(x_1) e^{-\frac{x_2^2}{2}} \matrice{1 \\ -1}: \ f \in \SSS(\R)\right\}.
\end{equation}
The right inclusion follows from a computation. To prove the left inclusion, we pick $u$ such that $L_{0,1} u = 0$. We take the Fourier transform in $x_1$: this gives $L_{0,1}(\xi) \widehat{u}=0$, where
\begin{equation}
    L_{0,1}(\xi) = \xi \matrice{1 & 1 \\ 1 &  1}  + \matrice{ x_2 &  - iD_{x_2} \\  iD_{x_2} &   - x_2}.
\end{equation}
We fix $\xi$. The operator $L_{0,1}(\xi)$ is a linear differential operator; hence the space of decaying solutions to $L_{0,1}(\xi) v = 0$ is at most one-dimensional. Indeed, if $v_1, v_2$ are such functions, then their Wronskian is constant; and they decay. Thus their Wronskian vanishes; this implies that $v_1, v_2$ are linearly dependent. We then observe that
\begin{equation}
    L_{0,1}(\xi) e^{-\frac{x_2^2}{2}} \matrice{1 \\ -1} = 0.
\end{equation}
This shows that the kernel of $L_{0,1}(\xi)$ is one-dimensional. Superposing over $\xi$ yields \eqref{eq:1m}. Applying the equivalence between $L_{0,1}$ and $L_{\te,r}$, we conclude that the kernel of $L_{\te,r}$ is precisely made of functions
 \begin{equation}
   S_r \RR_\te U_\te  \left(f(x_1) e^{-\frac{x_2^2}{2}} \matrice{1 \\ -1}\right) = f\big(\sqrt{r}(R_\te x)_1\big) e^{-\frac{r(R_{\te}x_2)^2}{2}} \matrice{e^{-i\te/2} \\ -e^{i\te/2} }, \ \ \ \ f\in \SSS(\R).
 \end{equation}
This corresponds to \eqref{eq:1j}, where we rescaled $f$ by $\sqrt{r}$ (this preserves the Schwartz class). 
\end{proof}

\medskip

We define the space 
\begin{equation}
    \SSS_{\te,r}(\R^2) = \left\{ u \in \SSS(\R^2,\C^2) : \ u \in \ker_{\SSS(\R^2,\C^2)}(L_{\te,r})^\perp \right\},
\end{equation}
with orthogonality computed with respect to the $L^2$-scalar product. We provide $\SSS_{\te,r}(\R^2)$ with the seminorms inherited from $\SSS(\R^2,\C^2)$. 

\begin{lemm}\label{lem:1b} For every $\te \in \R$ and $r> 0$, the operator $L_{\te,r}$ acting on $\SSS_{\te,r}(\R^2)$ is one to one, with inverse $L_{\te,r}^{-1}$ bounded on $\SSS_{\te,r}(\R^2)$. 
\end{lemm}

\begin{proof} \textbf{1.} We recall that $L_{\te,r}$ and $\sqrt{r}L_{0,1}$ are conjugated by operators bounded on $\SSS(\R^2,\C^2)$, see \eqref{eq:2k}. Thus, it suffices to prove the lemma for $L_{0,1}$ only.

We introduce the annihilation and creation operators $\aaa$ and $\aaa^*$, as well as its associated quantum harmonic oscillator $\hh = \aaa^* \aaa$ and quantum states $\vp_n$:
\begin{align}
    \aaa = x_2 + \p_{x_2}, \ \ \ \ \aaa^* = x_2-\p_{x_2}, \ \ \ \ \mathfrak{h} = -\p_{x_2}^2 + x_2^2 -1,
    \\
    \vp_0(x_2) = \frac{1}{\pi^{1/4}}e^{-\frac{x_2^2}{2}}, \ \ \ \ \vp_n(x_2) = \dfrac{(\aaa^*)^n}{2^{n/2}\sqrt{n!}} \vp_0(x_2). 
\end{align}
The quantum states $\vp_n$ form a complete orthonormal basis of eigenvectors of $\mathfrak{h}$: for every $n$, $\|\vp_n\|_{L^2} = 1$ and $\mathfrak{h} \vp_n = 2n \vp_n$. Moreover they satisfy the creation and annihilation relations: $\aaa \vp_0 = 0$ and for $n\in\N$,
 \begin{equation}
 \label{eq:creationandann}
    \aaa^* \vp_n =  \sqrt{2n+2} \vp_{n+1}, \ \ \ \ \aaa \vp_{n+1} = \sqrt{2n+2} \vp_n.
\end{equation}
Introduce
\begin{equation}\label{eq:7j}
  \tL_{0,1} = \matrice{1 & -1 \\ 1 &  1} L_{0,1} \matrice{1 & -1 \\ 1 &  1}^{-1} =  \matrice{0 & \aaa^* \\ \aaa & 2D_{x_1}},
\end{equation}
and the associated space $\tSSS_{0,1}(\R^2)$ -- defined similarly as $\SSS_{0,1}(\R^2)$:
\begin{align}
    \tSSS_{0,1}(\R^2) & = \left\{ u \in \SSS(\R^2) : \ u \in \ker_{\SSS(\R^2)}(\tL_{0,1})^\perp \right\} 
    \\
    & =  \left\{ u \in \SSS(\R^2) : \  \forall x_1 \in \R, \  \int_{\R^2} u_1(x) \vp_0(x_2) dx_2 = 0 \right\}. \label{eq:8c}
\end{align}
The lemma boils down to prove that $\tL_{0,1}$ is invertible on $\tSSS_{0,1}(\R^2)$. 

\newcommand{\WW}{\mathcal{W}}
\textbf{2.} Let $\WW$ be the Fr\'echet space of functions $w \in C^\infty(\R \times \N, \C^2)$ such that $w_1(\cdot,0) = 0$, equipped with the seminorms 
\begin{equation}
    N_{\alpha,\beta,\gamma}(w) = \sup_{n, \xi} \left| \lr{n}^{2\alpha}  \lr{\xi}^\beta \p_{\xi}^\gamma w(\xi,n) \right|, \ \ \ \ \alpha, \beta, \gamma \in \N.
\end{equation}
We define $S : \tSSS_{0,1}(\R^2) \rightarrow \WW$ by
\begin{equation} \label{eq:S}
    Su(\xi,n) = \int_{\R^2} e^{-i\xi x_1} \matrice{u_1(x) \vp_{n+1}(x_2) \\ u_2(x) \vp_n(x_2)} dx, \ \ \ \ u \in \tSSS_{0,1}\big(\R^2\big), \ n \in \N, \ \xi \in \R. 
    \end{equation}

We first observe that $S : \tSSS_{0,1}(\R^2) \rightarrow \WW$ is continuous. Indeed, if $u \in \tSSS_{0,1}(\R^2)$ and $\alpha, \beta, \gamma \in \N$, we have
\begin{equation}\label{eq:2e}
    \lr{2n}^{2\alpha}  \lr{\xi}^\beta D_\xi^\gamma Su(\xi,n)= S v(\xi,n), \ \ \ \ v(x) = \lr {\hh}^{2\alpha} \lr{D_{x_1}}^\beta (-x_1)^\gamma u(x).
\end{equation}
Moreover, $v\in \SSS(\R^2)$ when $u \in \SSS(\R^2)$. The Cauchy--Schwarz inequality yields
\begin{equation}
 N_{\alpha,\beta,\gamma}(Su) = \sup_{n,\xi} \big| Sv(\xi,n) \big|\leq \sup_n \int_{\R^2} \left| \matrice{v_1(x) \vp_{n+1}(x_2) \\ v_2(x) \vp_n(x_2)} \right| dx \leq 2\int_\R \left(\int_\R \big| v(x) \big|^2 dx_2\right)^{1/2} dx_1,
\end{equation}
where we used  $\|\vp_n\|_{L^2} = 1$.
The RHS is controlled by Schwartz semi-norms of $v=\lr{\hh}^{2\alpha} \lr{x_1}^\beta D_{x_1}^\gamma u$, thus of $u$. Hence $u \equiv 0$ and $S$ is continuous. 

Moreover, $S$ is invertible. The range of $S$ is $\WW$: if $w \in \WW$ then we have $Su = w$ with
\begin{equation}
     u(x) = \frac{1}{2\pi} \int_\R  e^{i\xi x_1} \sum_{n=0}^\infty \matrice{ \vp_{n+1}(x_2) w_1(\xi,n) \\ \vp_n(x_2) w_2(\xi,n)} d\xi,
\end{equation}
using the Fourier inversion formula and orthogonality relations for the $\vp_n$. We now show that $S$ is one-to-one. If $u \in \tSSS_{0,1}(\R^2)$ is such that $Su \equiv 0$ then
\begin{equation}\label{eq:8d}
    \forall x_1 \in \R, \ n \in \N, \ \ \ \ \int_\R \matrice{u_1(x) \vp_{n+1}(x_2) \\ u_2(x) \vp_n(x_2)} dx_2 = 0
\end{equation}
from the Fourier inversion formula. Since $\vp_n$ forms an orthonormal basis of $L^2(\R)$, \eqref{eq:8d} implies that $u_2 \equiv 0$ and $u_1(x) = c(x_1) \vp_0(x_2)$. From $u \in \tSSS_{0,1}(\R^2)$ and  \eqref{eq:8c}, $u_1 \equiv 0$. Hence $S$ is invertible.

\textbf{3.} Because of the closed graph theorem, invertible continuous operators between Fr\'echet spaces have continuous inverses. Hence the inverse of $S$ is continuous from $\WW$ to $\tSSS_{0,1}(\R^2)$.  Hence, to prove the lemma it suffices to show that $S \tL_{0,1} S^{-1} : \WW \rightarrow \WW$ is continuously invertible. But $S \tL_{0,1} S^{-1}$ is actually a simple multiplication operator: using that $D_{x_1}$ corresponds to $\xi$ in Fourier space and $\aaa, \aaa^*$ are shift operators -- see \eqref{eq:creationandann} --  in Hermite space, we have:
% Fourier inversion formula, the properties of Hermite functions, and an integration by parts, we have
\begin{equation}\label{eq:1y}
 S \tL_{0,1} S^{-1} w(\xi,n) =   \matrice{0 & \sqrt{2n+2} \\ \sqrt{2n+2} & 2\xi} w(\xi,n).
\end{equation}
This is a continuous operator on $\WW$; and \eqref{eq:1y} yields a formula for $\tL_{0,1}^{-1}$:
\begin{equation}\label{eq:4p}
  \tL_{0,1}^{-1} 
 =  S^{-1} \dfrac{1}{2n+2} \matrice{2\xi & -\sqrt{2n+2} \\ -\sqrt{2n+2} & 0} S.
\end{equation}
This completes the proof.
\end{proof}

\subsection{Solving the dominant equation.}\label{sec:3.2} We now focus on solving the hierarchy of equations~\eqref{eq:3g}, starting with the first two:
\begin{equation}
    T_0 a_0 = 0, \ \ \ \ T_0a_1 + T_1 a_0 = 0.
\end{equation}
Below we abuse notation: we allow functions in $\SSS(\R)$ or $\SSS(\R^2,\C^2)$ to also depend smoothly on time, and we consider the operator $\KK_t$ from \eqref{def:a_0} on functions depending on $t$. For instance, we write \eqref{eq:3n} as
\begin{equation}\label{eq:3n}
    a_0(t,x) = \KK_t f_0(t,x)  = r_t^{1/4} f_0\big(t,(R_{\te_t} x)_1\big) e^{-\frac{r_t (R_{\te_t}x)_2^2}{2}} \matrice{e^{-i\te_t/2} \\ -e^{i\te_t/2} }
\end{equation}

Since $T_0 = L_{\te_t,r_t}$, Lemma \ref{prop:1} implies that for any $f_0 \in \SSS(\R)$ (potentially depending on $t$), \eqref{eq:3n}
solves the equation $T_0 a_0 = 0$. 

%While the factor $r_t^{1/4}$ could be incorporated in the time-dependence of $f_0$, it will simplify technical aspects below by considering it separately. One can regard it as a renormalization factor: up to the volume-preserving substitution $x \mapsto R_{\te_t}x$, the expression \eqref{eq:3n} involves the function $x_2 \mapsto r_t^{1/4} e^{-\frac{r_t x_2^2}{2}}$ which has $L^2(\R)$-norm that is constant in time. 

\subsection{Solving the subleading equation.}\label{sec:5.3} The subleading equation in the hierarchy \eqref{eq:3g} is $T_0 a_1 + T_1 a_0 = 0$ where $T_0=L_{\te_t,r_t}$ and 
\begin{equation}\label{eq:4v}
    T_1 =  D_t + \sum_{|\alpha| = 2} \dfrac{\p^\alpha \kappa(y_t)}{\alpha!}  x^\alpha \sigma_3.
\end{equation}
Given $a_0$ satisfying \eqref{eq:3n}, we regard $T_0 a_1 + T_1 a_0 = 0$ as an equation with unknown $a_1 \in \SSS(\R^2,\C^2)$. According to Lemma \ref{lem:1b}, a solution exists if for any $t \in \R$, $T_1 a_0(t,\cdot) \in \SSS_{\te_t,r_t}(\R^2)$. We now look for $f_0$ such that this holds.

\medskip 

We note that $T_1 a_0 \in \SSS_{\te_t,r_t}(\R^2)$ if and only if for every $t \in \mathbb R$ and $g\in \SSS(\mathbb R)$:
\begin{equation}\label{eq:3l}
    \int_{\R^2} g\big((R_{\te_t} x)_1\big) e^{-\frac{r_t (R_{\te_t}x)_2^2}{2}} \matrice{e^{i\te_t/2} \\ -e^{-i\te_t/2}} \cdot T_1 a_0(t,x) dx = 0.
\end{equation}
We make the substitution $x \mapsto R_{\te_t}^\top x$ and pick functions $g$ approaching delta distributions to obtain that \eqref{eq:3l} is equivalent to:
\begin{equation}\label{eq:1l}
\forall t, x_1 \in \R, \ \ \ \ 
    \int_{\R}   e^{-\frac{r_t x_2^2}{2}} \matrice{e^{i\te_t/2} \\ -e^{-i\te_t/2}} \cdot (T_1 a_0)\big(t,R_{\te_t}^\top x\big) dx_2 = 0.
\end{equation}

\begin{lemm}\label{lem:1d} If $f(t,\cdot) \in \SSS(\R)$ depends smoothly on $t$, then
\begin{equation}\label{eq:4h}
    \int_{\R}   e^{-\frac{r_t x_2^2}{2}} \matrice{e^{i\te_t/2} \\ -e^{-i\te_t/2}} \cdot (T_1 \KK_t f)\big(t,R_\te^\top x\big) dx_2 = 2 \sqrt{\dfrac{\pi}{r_t}} D_t f(t,x_1).
\end{equation}
\end{lemm}
 
\begin{proof} We note the identities
\begin{equation}\label{eq:1o}
  \lr{ \matrice{e^{-i\te_t/2} \\ -e^{i\te_t/2}},  \sigma_3 \matrice{e^{-i\te_t/2} \\ -e^{i\te_t/2}} } = 0, \ \ \ \ \lr{ \matrice{e^{-i\te_t/2} \\ -e^{i\te_t/2}},   \matrice{-\dot{\te_t} e^{-i\te_t/2} \\ -\dot{\te_t} e^{i\te_t/2}} } = 0.
\end{equation}
Therefore, using the expressions \eqref{eq:4v} for $T_1$ and \eqref{def:a_0} for $\KK_t$,  we have:
\begin{multline}
    \matrice{e^{i\te_t/2} \\ -e^{-i\te_t/2}} \cdot T_1 \KK_tf(t,x) = 2D_t \left( r_t^{1/4} f\big(t,(R_{\te_t} x)_1\big) e^{-\frac{r_t (R_{\te_t}x)_2^2}{2}} \right)
 \\ =  \dfrac{2}{i} e^{-\frac{r_t (R_{\te_t}x)_2^2}{2}} \left( \dfrac{\p}{\p t}   + (\dot{R_{\te_t}} x)_1 \dfrac{\p}{\p x_1}  - \dfrac{\dot{r_t}(R_{\te_t}x)_2^2}{2}  - r_t  (R_{\te_t}x)_2 (\dot{R_{\te_t}}x)_2 \right) r_t^{1/4} f\big(t,(R_{\te_t} x)_1\big).  
\end{multline}
We deduce that 
\begin{multline}\label{eq:3m}
    \matrice{e^{i\te_t/2} \\ -e^{-i\te_t/2}} \cdot T_1 \KK_t f(t,R_{\te_t}^\top x)  
    \\
    =
    -2i e^{-\frac{r_t x_2^2}{2}} \left( \dfrac{\p}{\p t}  + (\dot{R_{\te_t}} R_{\te_t}^\top x)_1 \dfrac{\p}{\p x_1}  - \dfrac{\dot{r_t}x_2^2}{2} - r_t  x_2 (\dot{R_{\te_t}} R_{\te_t}^\top x)_2 \right) r_t^{1/4}   f\big(t,x_1\big).
\end{multline}

We remark that 
\begin{equation}\label{eq:7u}
\dot{R_{\te_t}} \cdot R_\te^\top  x =  \dot{\te_t} \matrice{0 & 1 \\ -1 & 0} x = \dot{\te_t} \matrice{x_2 \\ -x_1}.
\end{equation}
We deduce that \eqref{eq:3m} becomes:  
\begin{align}
    \matrice{e^{i\te_t/2} \\ -e^{-i\te_t/2}} T_1 \KK_t f(t,R_{\te_t}^\top x)
   =
    -2i  e^{-\frac{r_t x_2^2}{2}} \left( \dfrac{\p}{\p t}   + \dot{\te_t} x_2  \dfrac{\p}{\p x_1}  - \dfrac{\dot{r_t}x_2^2}{2} + r_t \dot{\te_t}  x_2 x_1 \right)  r_t^{1/4} f\big(t,x_1\big).
\end{align}
We plug this identity in \eqref{eq:4h} to obtain:
\begin{equation}\label{eq:4y}
   -2i \int_{\R}   e^{-r_t x_2^2} \left( \dfrac{\p}{\p t}  + \dot{\te_t} x_2  \dfrac{\p}{\p x_1}  - \dfrac{\dot{r_t}x_2^2}{2} + r_t \dot{\te_t}  x_2 x_1 \right) dx_2 \cdot r_t^{1/4}  f\big(t,x_1\big).
\end{equation}
We now perform the integrals over $x_2$. The function $x_2 e^{-r_t x_2^2}$ has vanishing integral; moreover an integration by parts shows that
\begin{equation}
   \sqrt{\dfrac{\pi}{r_t}} =   \int_{\R}   e^{-r_t x_2^2} dx_2 = 2r_t \cdot \int_{\R}x_2^2  e^{-r_t x_2^2} d x_2.
\end{equation}
Hence \eqref{eq:4y} reduces to:
\begin{equation}\label{eq:4z}
    -2i \sqrt{\dfrac{\pi}{r_t}}\left( \dfrac{\p}{\p t} - \dfrac{\dot{r_t}}{4r_t} \right) r_t^{1/4} f\big(t,x_1\big). 
\end{equation}

We finally observe that in the sense of differential operators,
\begin{equation}
    \left( \dfrac{\p}{\p t} - \dfrac{\dot{r_t}}{4r_t} \right) r_t^{1/4} = \dfrac{\p}{\p t}.
\end{equation}
Using this identity in \eqref{eq:4z} completes the proof. 
\end{proof}

From \eqref{eq:1l} and Lemma \ref{lem:1d}, we obtain the transport equation for $f_0$: $D_t f_0 = 0$. Hence,~$f_0$ depends on $x_1$ only, and we write $f_0(t,x_1) = f_0(x_1)$. Therefore, if
\begin{equation}\label{eq:3k}
a_0(t,x) = r_t^{1/4} f_0\big((R_{\te_t} x)_1\big) e^{-\frac{r_t (R_{\te_t}x)_2^2}{2}} \matrice{e^{-i\te_t/2} \\ -e^{i\te_t/2} } = \KK_t f_0(x)
\end{equation} 
for some $f_0 \in \SSS(\R)$, then $T_1 a_0(t,\cdot) \in \SSS_{\te_t,r_t}(\R^2)$ for every $t \in \R$; hence the equation $T_0 b_1 + T_1 a_0 = 0$ has a unique solution $b_1$ such that $b_1(t,\cdot) \in \SSS_{\te_t,r_t}(\R^2)$ for every $t \in \R$. We obtain the general solution to $T_0 a_1 + T_1 a_0 = 0$ by adding an element of $\ker(L_{\te_t,r_t})$: $a_1 = b_1 + \KK_t  f_1$:
\begin{equation}\label{eq:1r}
    a_1(t,x) = b_1(t,x) + r_t^{1/4} f_1\big(t,(R_{\te_t}x)_1\big) e^{-\frac{r_t (R_{\te_t} x)^2_2}{2}} \matrice{e^{-i\te_t/2} \\ -e^{i\te_t/2} }, \ \ \ \ f_1(t,\cdot) \in \SSS(\R).
\end{equation}

\subsection{Proof of Theorem \ref{thm:2}}\label{sec:3.6}

We are now in a position to prove Theorem \ref{thm:2}. We start with a classical result based on Duhamel's formula.

\begin{lemm}\label{lem:1a} Let $\psi_t \in \SSS(\R^2)$ be a solution to $(\epsi D_t + H) \psi_t = 0$. Then for any $v_t \in \SSS(\R^2)$,
\begin{equation}
    \big\| v_t - \psi_t \big\|_{L^2} \leq \| v_0 -\psi_0 \|_{L^2} + \dfrac{1}{\epsi} \int_0^t \big\| (\epsi D_s + H)v_s \big\|_{L^2} ds.
\end{equation}
\end{lemm}

\begin{proof} Let $w_t = v_t - \psi_t$ and $r_t = (\epsi D_t + H)v_t$. Then,  $(\epsi D_t + H)w_t = r_t$. 
By Duhamel's formula, 
\begin{equation}
  v_t - \psi_t =  w_t = e^{-itH/\epsi} w_0 + \dfrac{1}{\epsi} \int_0^t e^{-i(t-s) H / \epsi} r_s ds = e^{-itH/\epsi} (v_0-\psi_0) + \dfrac{1}{\epsi} \int_0^t e^{-i(t-s) H / \epsi} r_s ds.
\end{equation}
We bound both sides in $L^2$, using that $e^{-itH}$ is unitary:
\begin{equation}
    \| v_t -\psi_t \|_{L^2} \leq  \| v_0 -\psi_0 \|_{L^2} + \dfrac{1}{\epsi} \int_0^t \big\| (\epsi D_s + H)v_s \big\|_{L^2} ds.
\end{equation}
This completes the proof.
\end{proof} 

\begin{proof}[Proof of Theorem \ref{thm:2}] \textbf{1.} Let $f_0 \in \mathcal{S}(\field{R})$. Let $a_0$ as in \eqref{eq:3n}, $b_1$ is as in \eqref{eq:1r} and $a^{(1)} = a_0+\epsi^{1/2} a_1$. We apply Lemma \ref{lem:1e} with $m = 1$:
\begin{equation}\label{eq:2o}
   \big\| (\epsi D_t + H) W[a^{(1)}]_{y_t} \big\|_{L^2} \leq C \epsi^{3/2} \left( \| D_t b_1 \|_{L^2} + \big\| \lr{x}^3 a_0 \big\|_{L^2} + \big\| \lr{x}^3 b_1 \big\|_{L^2} \right).
\end{equation}

\textbf{2.} We now bound the right-hand-side of \eqref{eq:2o}, starting with $\lr{x}^3 a_0$ in $L^2$. We write $a_0 = \KK_{\te_t,r_t} f_0$, 
where
\begin{equation}
    \KK_{\te,r} = r^{1/4} f\big( (R_\te x)_1 \big) e^{-\frac{r(R_\te x)_2^2}{2}} \matrice{e^{-i\te/2} \\ - e^{i\te/2}}.
\end{equation}
We note that we have the identity $\KK_{\te,r} = \DD_r \,\UU_\te \,\KK_{0,1}$, where $\DD_r$ is a partial dilation operator and $\UU_\te$ was introduced in  \eqref{eq:3b}:
\begin{equation}\label{eq:2q}
\DD_r g(x) = r^{1/4} g\left(x_1, \sqrt{r}x_2\right), \ \ \  \ \UU_\te g(x) = \matrice{e^{-i\te/2} & 0 \\ 0 & e^{i\te/2}} g(R_\te x).
\end{equation}
The operator $\KK_{0,1}$ is bounded from $\SSS(\R)$ to $\SSS(\R^2,\C^2)$; $\UU_\te$ is uniformly bounded from  $\SSS(\R^2)$  to $\SSS(\R^2,\C^2)$ for $\te \in \R$; and  $\DD_r$ is bounded uniformly on $\SSS(\R^2)$ for $r$ in compact subsets of $(0,\infty)$. Moreover, $r_t = \big| \nabla \kappa(y_t) \big|$ lives in a compact subset of $(0,\infty)$, because of  $\kappa \in C^\infty_b(\R^2)$ and \eqref{eq:7z}. We deduce that $a_0 \in \SSS(\R^2)$, with uniform-in-time bounds on its seminorms. In particular, $\big\| \lr{x}^3 a_0 \big\|_{L^2}$ is uniformly bounded. 

\medskip 

For later use, we observe that $\p_t a_0$ is also uniformly bounded in $\SSS(\R^2)$. Indeed, from \eqref{eq:2q}, we have 
\begin{equation}\label{eq:2r}
    \p_t a_0 =  \dot{r_t} \p_r \DD_{r_t} \, \UU_{\te_t}\, \KK_{0,1} f_0 + \dot{\te_t} \,\DD_{r_t} \,\p_\te \UU_{\te_t} \, \KK_{0,1} f_0.
\end{equation}
The operators $\p_\te  \UU_{\te_t}$ and $\p_r \DD_{r_t}$ are uniformly bounded on $\SSS(\R^2)$ -- the latter because $r_t$ lives in a compact subset of $(0,\infty)$. The quantities $\dot{r_t}$ and $\dot{\te_t}$ are uniformly bounded: 
\begin{equation}
    \left| \dot{r_t} \right| = \dfrac{\lr{\nabla \kappa(y_t), \nabla^2 \kappa(y_t)}}{2 \big|\nabla \kappa (y_t)\big|} \leq \big| \nabla \kappa(y_t) \big| \cdot \big| \nabla^2 \kappa(y_t) \big| \leq C;
\end{equation}
and likewise,
\begin{equation}
  \big|\dot{\te_t}\big| = \left| \dfrac{d}{dt} \dfrac{\nabla \kappa(y_t)}{r_t} \right| \leq \dfrac{1}{r_t} + \dfrac{|\nabla^2 \kappa(y_t)|}{r_t} \leq C.
\end{equation}

Therefore, we deduce from \eqref{eq:2r} that $\p_t a_0$ is uniformly bounded in $\SSS(\R^2)$.  

\textbf{3.} We now control in $L^2$ the terms $D_t b_1$ and $\lr{x}^3 b_1$ that appear in \eqref{eq:2o}. We use \eqref{eq:2k} to write $b_1$ as:
\begin{equation}\label{eq:2s}
    b_1(t,\cdot) = -L_{\te_t,r_t}^{-1} a_0 = -\sqrt{r_t} S_{r_t}^{-1} \UU_{\te_t}^{-1} L_{0,1}^{-1} \UU_{\te_t} S_{r_t} a_0.
\end{equation}
As in Step 2, all operators involved in \eqref{eq:2s} are uniformly bounded in $\SSS(\R^2)$, and we deduce that $b_1 \in \SSS(\R^2)$ uniformly in time. Also similarly to \eqref{eq:2r}, taking time derivatives produces quantities such as $\dot{r_t}$, $r_t^{-1/2}$, $\dot{\te_t}$ (all uniformly bounded); operators such as $\p_r \DD_{r_t}, \p_r \DD_{r_t^{-1}}$, $\p_\te \UU_{\te_t}$ and $\p_\te \UU_{-\te_t}$, all uniformly bounded on $\SSS(\R^2)$; and the function $\p_t a_0$ -- also bounded uniformly in $\SSS(\R^2)$. We deduce that $b_1, \p_t b_1$ are uniformly in $\SSS(\R^2)$. Hence, $\big\| \lr{x}^3 b_1 \big\|_{L^2}$ and $\big\| \p_t b_1 \big\|_{L^2}$ are uniformly bounded. 

\textbf{4.} Going back to \eqref{eq:2o}, we have for any $t$:
\begin{equation}\label{eq:2v}
   \left\| (\epsi D_t + H) W[a^{(1)}]_{y_t} \right\|_{L^2} \leq C \epsi^{3/2}.
\end{equation}
Let $\psi_t$ be the solution to $(\epsi D_t + H)\psi_t = 0$ with initial data $\psi_0 = a_0(0,\cdot)$; and $v_t = W[a^{(1)}]_{y_t}$. We note that $v_0 - \psi_0 = \epsi^{1/2} b_1(0,\cdot)$ and that $v_t$ satisfies the bound \eqref{eq:2v}. Thanks to Lemma \ref{lem:1a}, we get
\begin{equation}
    \| v_t - \psi_t \|_{L^2} \leq \epsi^{1/2} \big\| b_1(0,\cdot) \big\|_{L^2} + C \epsi^{1/2} t.
\end{equation}
Therefore,
\begin{equation}
    \psi_t = W[a^{(1)}]_{y_t} + \Or_{L^2}\big(\epsi^{1/2} \lr{t}\big) = W[a_0]_{y_t} + \Or_{L^2}\big(\epsi^{1/2} \lr{t}\big).
\end{equation}
This completes the proof. \end{proof}

\subsection{Subsequent equations}\label{sec:3.7} We now focus on deriving a version of Theorem \ref{thm:2} that favors accuracy  over lifetime. This requires to solve higher-order transport equations.

The base case is the result of \S\ref{sec:3.2}-\ref{sec:5.3}, summarized as follows:
\begin{quote}
\Hone For any $f_0\in\SSS(\R)$, there exists $b_1$ such that for any $f_1(t,\cdot) \in \SSS(\R)$ if $a_0=\KK_t f_0$ and $a_1=b_1+\KK_t f_1$,  then $a_0$ and $a_1$ solve \eqref{eq:3g} with $m = 1$, i.e.
\begin{equation}
    \sum_{\ell = 0}^j T_{j-\ell} a_j = 0, \quad 0 \leq j \leq 1.
\end{equation}
\end{quote} 
To construct $a_0$ and $a_1$, we had to enforce a condition on $f_0$. Likewise, to construct $a_m$ we will enforce a condition on $f_{m-1}$.

Our inductive assumption is, for $m \geq 1$:
\begin{quote}
\Hm For any $f_0 \in \SSS(\R)$, there exist $b_1, f_1, \dots, b_{m-1}, f_{m-1}, b_{m} \in \SSS(\R)$ depending smoothly on $t$, such that for any $f_{m} \in \SSS(\R)$, if $a_0 = \KK_t f_0$ and $a_\ell = b_\ell + \KK_t f_\ell$ then 
\begin{equation}
    \sum_{\ell = 0}^j T_{j-\ell} a_j = 0, \quad 0 \leq j \leq m.
\end{equation}
\end{quote}

We proved \Hone in \S\ref{sec:5.3}.  We now assume that \Hmm holds and we prove \Hm for $m \geq 2$. Because of Lemma \ref{lem:1e}, this boils down to constructing $a_{m} = b_{m} + \KK_t f_{m}$ such that:
\begin{equation}\label{eq:2i}
    T_0 (b_{m} + \KK_t f_{m}) + T_1 a_{m-1} + \dots + T_{m} a_0 = 0, \ \ \ \ \text{where:}
\end{equation}
\begin{itemize}
    \item The operators $T_k$ are defined in \eqref{def:Tj};
    \item The amplitudes $a_0, \dots, a_{m-2}$ are fully specified by \Hmm \hspace{-1.2mm};
    \item The amplitude $a_{m-1} = b_{m-1} + \KK_t  f_{m-1}$, with $b_{m-1}$ given by \Hmm and $f_{m-1} \in \SSS(\R)$ remains be selected.
\end{itemize}

Since the operator $\KK_t $ parametrizes the kernel of $T_0$, \eqref{eq:2i} is equivalent to
 \begin{equation}\label{eq:4b}
     T_0 b_{m} =  \beta_{m-1} - T_1 \KK_t f_{m-1}, \ \ \ \  \beta_{m-1} = - T_1 b_{m-1} - T_2 a_{m-2} - \dots - T_{m} a_0. 
 \end{equation}
Note that \Hmm fully prescribes $\beta_{m-1}$. 
 
As in \S\ref{sec:5.3}, to solve \eqref{eq:4b}, it suffices that for any $t$, $\big(\beta_{m-1} - T_1 \KK_t f_{m-1}\big)(t,\cdot)$ is in the kernel of $T_0$. This is equivalent to
\begin{equation}
    \forall t, x_1 \in \R, \ \ \ \ 
    \int_{\R}   e^{-\frac{r_t x_2^2}{2}} \matrice{e^{i\te_t/2} \\ -e^{-i\te_t/2}} \cdot \big( \beta_{m-1} - T_1 \KK_t f_{m-1} \big)\big(t,R_{\te_t}^\top x\big) dx_2 = 0.
\end{equation}
Thanks to Lemma \ref{lem:1d}, this is equivalent to:
\begin{equation}\label{transport:fm}
D_t f_{m-1}(t,x_1)
=   \dfrac{1}{2}  \sqrt{ \dfrac{r_t}{\pi} }  \int_{\R}   e^{-\frac{r_t x_2^2}{2}} \matrice{e^{i\te_t/2} \\ -e^{-i\te_t/2}} \cdot \beta_{m-1}\big(t,R_{\te_t}^\top x\big) dx_2,
\end{equation}
and hence -- setting $f_{m-1}(0,x_1) = 0$:
\begin{equation}\label{eq:1q}
    f_{m-1}(t,x_1) =  \int_0^t \int_{\R} \dfrac{1}{2}  \sqrt{ \dfrac{r_s}{\pi} }     e^{-\frac{r_s x_2^2}{2}} \matrice{e^{i\te_s/2} \\ -e^{-i\te_s/2}} \cdot \beta_{m-1}\big(s,R_{\te_s}^\top x\big) dx_2 ds.
\end{equation}
When $f_{m-1}$ is given by this formula, the equation \eqref{eq:4b}  admits a solution $b_m(t,\cdot) \in \SSS(\R^2,\C^2)$. This completes the proof of \Hm \hspace{-1.2mm}. The following result summarizes our findings:

\begin{theo}\label{thm:4} Fix $T > 0$ and $n \in \N$. If $a_j \in \SSS(\R^2)$ are constructed as above, then $(\epsi D_t + H) \phi_t = 0$ 
%with initial data \begin{equation}     \phi_0(x) = \dfrac{1}{\sqrt{\epsi}} \cdot  \KK_0 f\left(\dfrac{x - y_t}{\sqrt{\epsi}} \right) + \sum_{j=1}^n \epsi^{\frac{j-1}{2} }  a_j\left(0, \dfrac{x - y_0}{\sqrt{\epsi}} \right) \end{equation}
has a solution of the form
\begin{equation} \label{eq:psi_with_correctors}
    \phi_t(x) = \dfrac{1}{\sqrt{\epsi}} \cdot  \KK_t f \left(\dfrac{x - y_t}{\sqrt{\epsi}} \right) + \sum_{j=1}^n \epsi^{\frac{j-1}{2} }  a_j\left(t, \dfrac{x - y_t}{\sqrt{\epsi}} \right)  + \Or_{L^2}\big(\epsi^{\frac{n+1}{2}} \big),
\end{equation}
uniformly for $\epsi \in (0,1]$ and $t$ in $[0,T]$.
\end{theo}

According to Theorem \ref{thm:4}, after adequately correcting the initial data \eqref{eq:initial_data} we obtain approximate solutions concentrated near $y_t$ at arbitrary accuracy in $\epsi$.
 Correcting the initial data is necessary: otherwise the subleading amplitude (which is of order $\epsi^{1/2}$) likely contains a dispersive part, hence cannot remain fully concentrated near $y_t$.

\begin{rem}[Timescale of validity of error estimates] Including higher order correctors as in  \eqref{eq:psi_with_correctors} does not extend the timescale of validity $\epsi^{-1/2}$ of the approximation solution. Indeed, the $n$-th corrector is of order $\epsi^{\frac{n+1}{2}} t^n$ -- the term $t^n$ corresponds to $n$ recursive integrations in \eqref{eq:1q}. After applying Lemma \ref{lem:1a}, this yields that the constant implicitly involved in the remainder $\Or_{L^2}\big(\epsi^{\frac{n+1}{2}} \big)$ of \eqref{eq:psi_with_correctors} grows like $T^{n+1}$: it is small only for $T \ll \epsi^{-1/2}$. 
\end{rem}

\begin{proof}[Proof of Theorem \ref{thm:4}] Fix $n \in \N$, $T > 0$ and $f_0 \in \mathcal{S}(\field{R})$. We pick $a_j$ solving \eqref{eq:3g} for $0 \leq j \leq n+1$ (constructed above) with $f_{n+1} = 0$, and we define
\begin{equation}\label{eq:v0}
    a^{(n)} = \sum_{j = 0}^{n+1} \epsi^{j/2} a_j, \ \ \ \ 
    v_t(x) = W\big[a^{(n)}\big]_{y_t}(x) = \frac{1}{\sqrt{\epsi}} \sum_{j = 0}^{n+1} \epsi^{j/2} a_j\left( t , \frac{x - y_t}{\sqrt{\epsi}} \right). 
\end{equation}
By construction, the functions $a_j$ are smooth in $t$ and Schwartz in $x$. In particular, they satisfy uniform Schwartz-class bounds for $t$ in compact intervals. Hence, thanks to Lemma \ref{lem:1e}, we have uniformly in $t \in [0,T]$:
\begin{equation}
   \left\| (\epsi D_t + H) v_t \right\|_{L^2} \leq C \epsi^{\frac{n+1}{2}}.
\end{equation}
Let $\phi_t$ be the solution to $(\epsi D_t + H) \phi_t = 0$ with $\phi_0 = v_0$ -- see \eqref{eq:v0} with $t=0$. 
Thanks to Lemma \ref{lem:1a}:
\begin{equation}
    \left\| v_t - \phi_t \right\|_{L^2} \leq C \epsi^{\frac{n+1}{2}}.
\end{equation}
In other words, $v_t = \phi_t + \Or_{L^2}\big( \epsi^{\frac{n+1}{2}}\big)$.
\end{proof}

\section{The effect of curvature}\label{sec:4}

It is natural to wonder which quantities affect the lifetime of our quantum state. For instance, when $\kappa$ is linear, the interface is straight and the edge states have infinite lifetime. If $\kappa$ is asymptotically linear, the interface is asymptotically straight and we expect an extended time of validity. In contrast, numerical simulations indicate that circular interfaces come with gradual dispersion: see Figure \ref{fig:7}.

This suggests that an integrated curvature limits the lifespan. Curvature however cannot be the only limiting factor: as Figure \ref{fig:1} shows, even straight interfaces can generate dispersion. To isolate the effects of curvature, we consider in this section domain walls $\kappa$ that satisfy a geometric condition:
\begin{equation}\label{eq:7s}
  y \in \kappa^{-1}(0) \ \ \ \Rightarrow \ \ \ \big| \nabla \kappa(y) \big| = 1, \ \ \ \  \nabla^2\kappa(y) \cdot \nabla \kappa(y) = 0.
\end{equation}
Example of $\kappa$ satisfying \eqref{eq:7s} include:
\begin{itemize}
    \item $\kappa(x) = \omega \cdot x$ with $|\omega| = 1$, for a straight interface;
    \item $\kappa(x) = \sqrt{x_1^2+x_2^2}-1$, for a circle.
\end{itemize}

The condition \eqref{eq:7s} is not geometrically restrictive: given $\Gamma$, we can always find $\kappa$ with $\Gamma = \kappa^{-1}(0)$, satisfying \eqref{eq:7s} -- see \S\ref{sec:4.2}. This condition excludes scenarios such as those giving rise to Figure \ref{fig:1}.  
Under \eqref{eq:7s}, $\dot{\te_t}$ is the curvature of $\Gamma$ at $y_t$; and in a suitable frame, the Hessian of $\kappa$ along $\Gamma$ depends only on $\dot{\te_t}$:
\begin{equation}\label{eq:7x}
   \lr{R_{\te_t}^\top x, \nabla^2 \kappa(y_t) R_{\te_t}^\top x} = \dot{\te_t} x_1^2.
\end{equation}

\begin{theo}\label{thm:3} Under \eqref{eq:7s}, the solution \eqref{eq:3s} to $(\eps D_t+H)\Psi_t=0$ of Theorem \ref{thm:1} satisfies, uniformly in $t > 0$ and $\epsi \in (0,1]$:
\begin{equation}\label{eq:0q}
\hspace{-3mm}
\Psi_t(x) = \dfrac{1}{\sqrt{\epsi}} \cdot \exp\left( - \dfrac{(x-y_t)^2}{2\epsi}\right) \matrice{ -e^{i\te_t/2} \\ e^{-i\te_t/2} } + \Or_{L^2}\big(\epsi^{1/2} + \epsi t(1+ \Theta_t)\big), \ \ \ \ \ \Theta_t = \int_0^t \dot{\te}_s^2 ds.
\end{equation} 
\end{theo}

When $\Gamma$ is asymptotically straight (i.e. it has $L^2$-curvature), the remainder in \eqref{eq:0q} remains small for $t \ll \epsi^{-1}$: our quantum state is longer-lived. In contrast, if $\Gamma$ is a closed loop then $\Theta_t$ grows linearly and our state is only close to the exact solution for $\epsi t^2 \ll 1$, that is $t \ll \epsi^{-1/2}$: there is no improvement over Theorem \ref{thm:1}. Thus, such states -- which are not globally topological -- have a  shorter lifetime. 

\iffalse 

\begin{figure}[t]
\floatbox[{\capbeside\thisfloatsetup{capbesideposition={right,center},capbesidewidth=3.5in}}]{figure}[\FBwidth]
{\hspace{-1cm}\caption{\label{fig:my_label}
Difference between the numerically computed solution to $(\epsi D_t + H )\Psi_t = 0$  and the Gaussian state \eqref{eq:0q} for a circle-type interface. This difference grows as time goes.
}}
{\begin{tikzpicture}
   \node at (0,0) {\includegraphics[width=6cm, height = 5cm]{alexis5.pdf}};
  \end{tikzpicture}}
\end{figure}

\fi

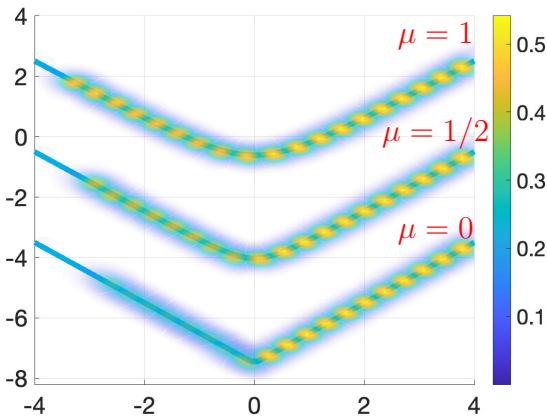
\begin{figure}[b]
\floatbox[{\capbeside\thisfloatsetup{capbesideposition={right,center},capbesidewidth=3.7in}}]{figure}[\FBwidth]
{\hspace{-1cm}\caption{\label{fig:8} Snapshots of the numerically computed solution to $(\epsi D_t +H)\Psi_t =0$, with $\epsi = 0.1$, $\Psi_0$ the  Gaussian \eqref{eq:3s} and a domain wall $\kappa$, illustrated in the figure with an appropriate off-set, satisfying \eqref{eq:7s}, with $\Gamma = \{x_2 = -\sqrt{x_1^2 + \mu^2}\}$, $\mu \in \{0,1/2,1\}$. We observe  a growing amplitude loss as the corner gets sharper. 
%The interface is described by  $\kappa(x) = (x_2-\sqrt{x_1^2 + \mu^2})f_1(x_1)+  (x_2-\sqrt{x_1^2 + \mu^2})^2 f_2(x_1)$ with $f_1(x_1)=\sqrt{ \frac{x_1^2+\mu^2}{2x_1^2+\mu^2}}$ and $f_2(x_1)=\frac{-x_1^2 \mu^2}{2(2x_1^2+\mu^2)^{5/2}}$ with $\mu \in \{0,0.5,1\}.$ 
}} 
{\begin{tikzpicture}
   \node at (0,0) {\includegraphics[width = 8cm]{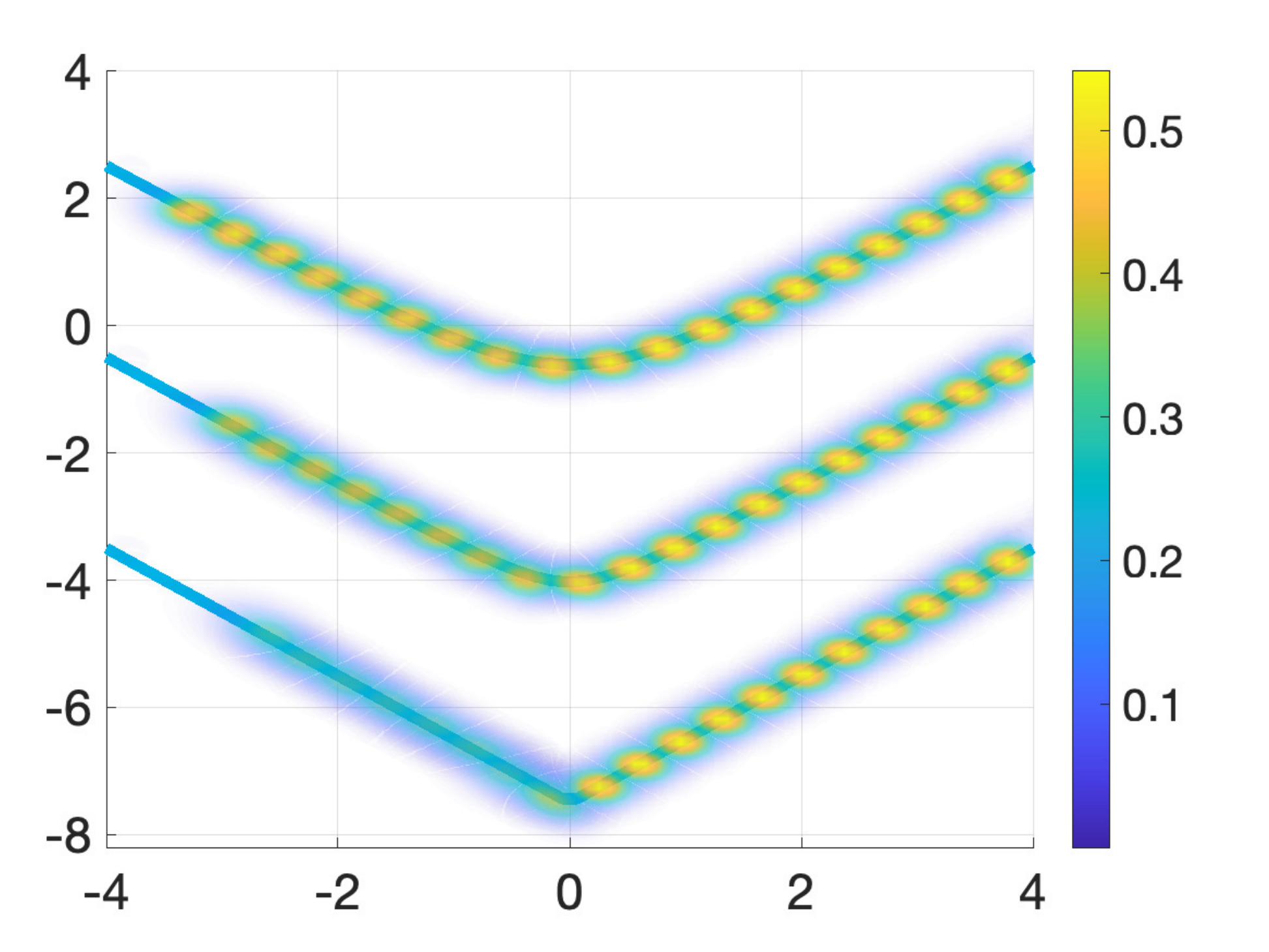}};
 \node[red] at (2,2.3) {$\mu=1$};
 \node[red] at (2,) {$\mu=1/2$};
 \node[red] at (2,-.3) {$\mu=0$};
  \end{tikzpicture}}
\end{figure}

Theorem \ref{thm:3} highlights effective limitations of dynamical edge states: they do not survive in strongly curved environments; see Figure \ref{fig:8}. This means that our results rely on $\kappa$ being sufficiently regular.  Other limitations include cross-type or knot-type interfaces, for which $\kappa$ degenerates quadratically; see Figure \ref{fig:my_crossing}. Such scenarios form interesting open problems.

\begin{figure}[t]
{\begin{tikzpicture}
   \node at (-4,0) {\includegraphics[width = 7cm]{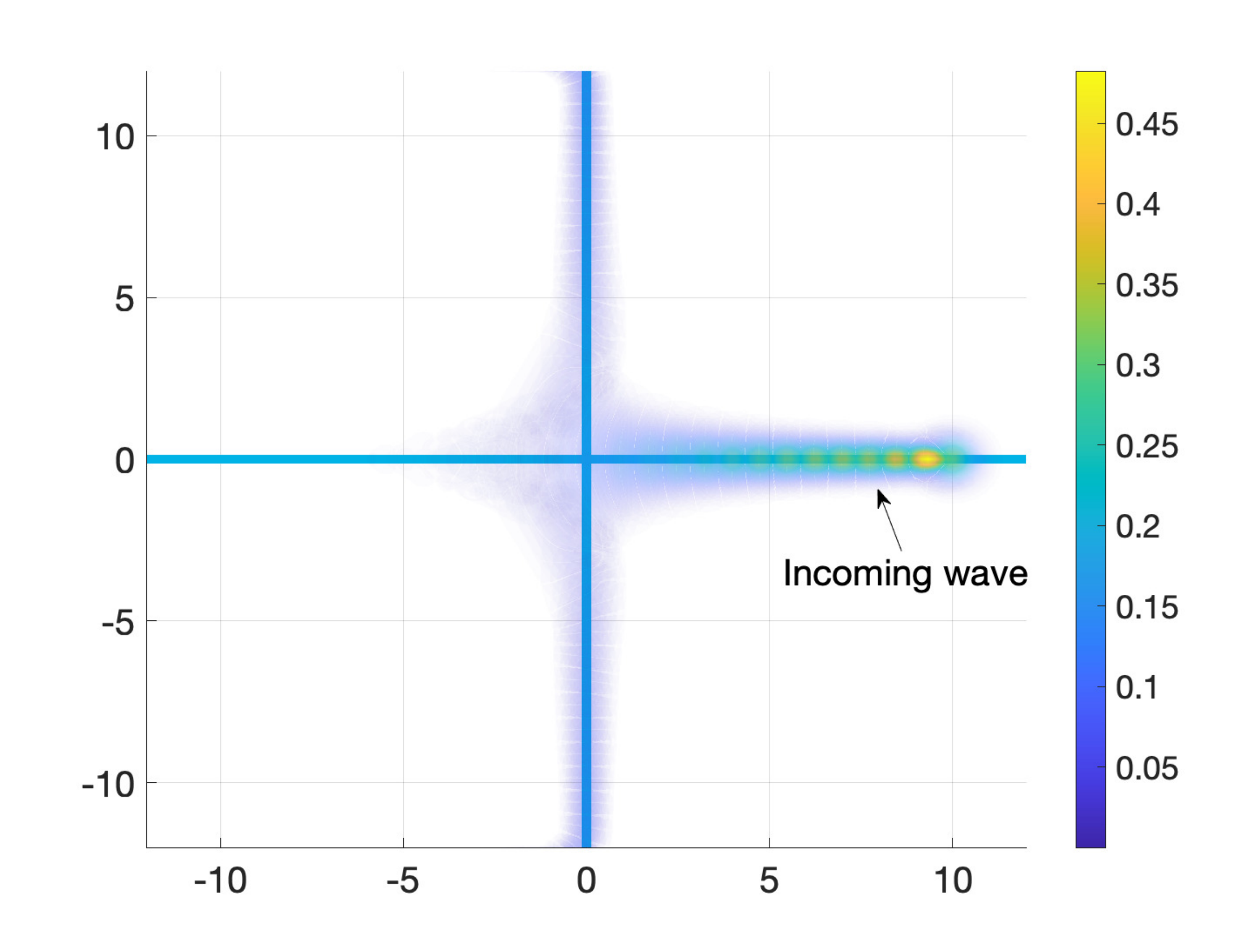}};
 \node at (4,0) {\includegraphics[width = 7cm]{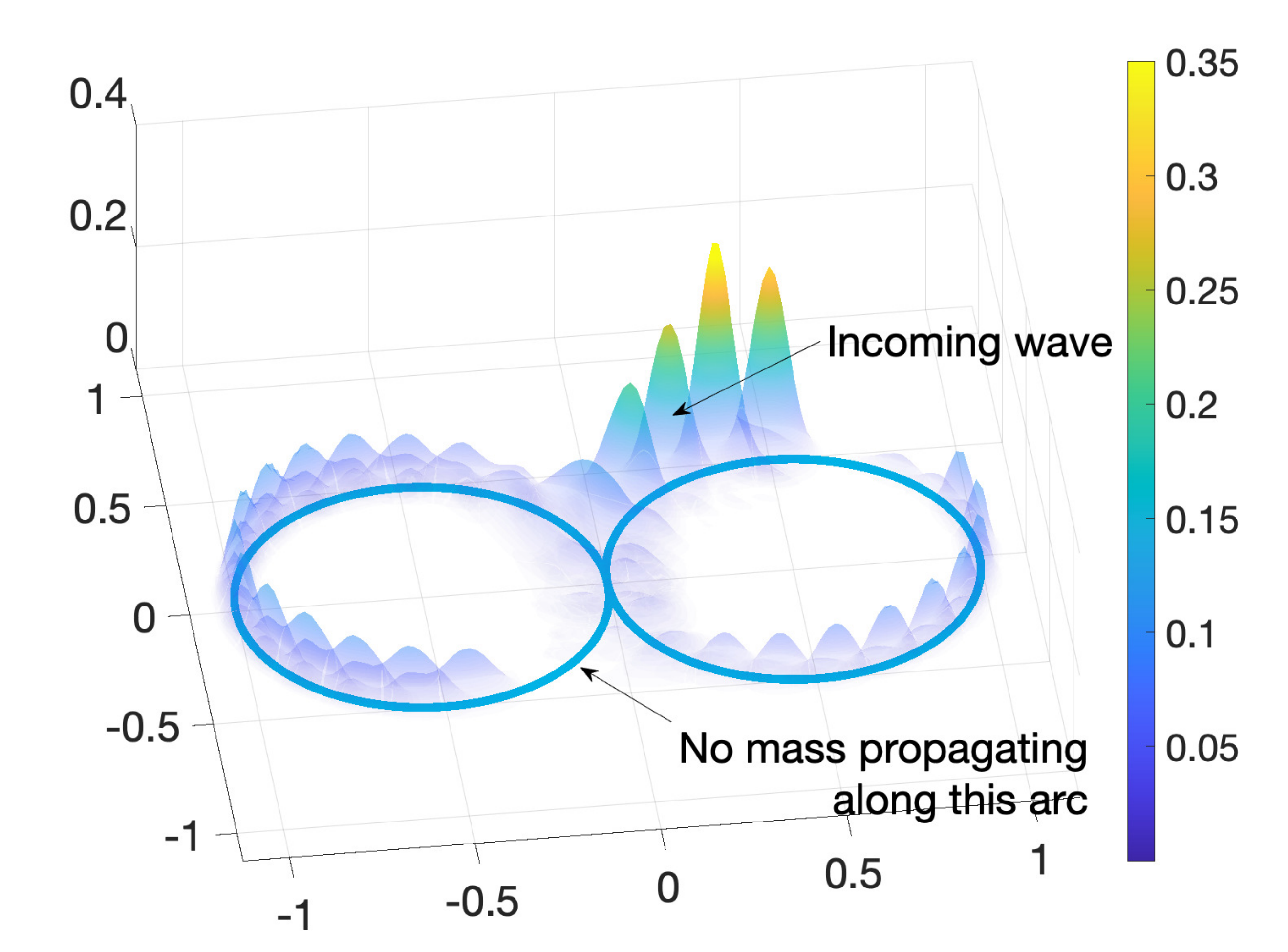}};
  \end{tikzpicture}}
  \caption{\label{fig:my_crossing}Left: interface $\kappa(x)=x_1x_2$; right: an interface consisting of two rings parametrized by $\vert x + e_1 \vert \vert x - e_1 \vert =1$ both with $\epsi = 2\cdot 10^{-2}$. While the direction of propagation can be heuristically predicted using the bulk-edge correspondence, establishing a rigorous theory remains an open problem.} 
\end{figure}

\subsection{Proof of Theorem \ref{thm:3}} The proof of Theorem \ref{thm:3} relies on the
precise calculation of the corrector $a_1=b_1+\KK_t f_1$  involved in \S\ref{sec:3.2}-\ref{sec:5.3}.

 \begin{lemm}\label{lem:1h} In the setup of Theorem \ref{thm:3}, the subleading amplitude $a_1 = b_1 + \KK_t f_1$ 
satisfies
\begin{equation}\label{eq:0j}
    b_1(t,x) =   \dfrac{1-x_1^2 }{2}  x_2 e^{-\frac{x^2}{2}}  \matrice{e^{-i\te_t/2} \\ - e^{i\te_t/2}} \dot{\te_t},
    \ \ \ 
    f_1(t,x_1) = 
    \frac{ 2x_1-x_1^3}{2} e^{-\frac{x_1^2}{2}  } \Theta_t 
\end{equation} 
\end{lemm}

\begin{proof}[Proof of Lemma \ref{lem:1h}] The proof relies on the hierarchy of transport equations studied in \S\ref{sec:3.2}. We use the notations introduced there, keeping in mind that $r_t=1$ here.

We first compute $b_1$. From the initial condition \eqref{eq:0r},
\begin{equation}
    a_0(0,x) = e^{-\frac{x^2}{2}} \matrice{e^{-i\te_0/2} \\ -e^{i\te_0/2}}.
\end{equation}
Hence $f_0(x_1)= e^{-x_1^2/2}$. Moreover $b_1$ is the unique solution in $\ker(T_0)^\perp$ to $T_0 b_1 + T_1 a_0 = 0$. With 
$q_t(x) = \lr{x, \nabla^2\kappa(y_t) x}$,
this equation reads
\begin{equation}\label{eq:7r}
    T_0 b_1 = - e^{-\frac{x^2}{2}} \left( D_t + \dfrac{q_t(x)}{2} \sigma_3 \right)  \matrice{e^{-i\te_t/2} \\ - e^{i\te_t/2}} = \dfrac{1}{2}e^{-\frac{x^2}{2}} \left( \dot{\te_t} - q_t(x) \right) \matrice{e^{-i\te_t/2} \\ e^{i\te_t/2}},
\end{equation}
where we used the identities \eqref{eq:1o}. 
To find $b_1$, we use the operators $U_{\te_t}$ and $\RR_{\te_t}$ introduced in \eqref{eq:3b} and we look for $b_1$ of the form 
\[ b_1 =  \RR_{\te_t} U_{\te_t} \matrice{c_1 \\ c_2}.\]
We take advantage of the relation $T_0=L_{\theta_t,1} = \RR_{\te_t} U_{\te_t} L_{0,1} U_{\te_t}^{-1} \RR_{\te_t}^{-1}$ (see \eqref{eq:Lte} and the beginning of the proof of Lemma~\ref{prop:1}) and 
apply the operator $U_{\te_t}^{-1} \RR_{\te_t}^{-1}$ to the equation \eqref{eq:7r}.  We deduce that $c_1$ and $c_2$ must solve:
\begin{equation}
    L_{0,1} \matrice{c_1 \\ c_2} =   \dfrac{1}{2}e^{-\frac{x^2}{2}} \left( \dot{\te_t} - q_t\left(R_{\te_t}^\top x\right) \right) U_{\te_t}^{-1} \matrice{e^{-i\te_t/2} \\ e^{i\te_t/2}} =  \dfrac{1}{2}e^{-\frac{x^2}{2}} \left( \dot{\te_t} - q_t\left(R_{\te_t}^\top x\right) \right) \matrice{1 \\ 1}. 
\end{equation}
We now use the operator $\tL_{0,1}$ of \eqref{eq:7j} and get:
\begin{equation}
     \matrice{0 & \aaa^* \\ \aaa & 2D_{x_1}} \matrice{c_1-c_2 \\ c_1+c_2} =  e^{-\frac{x^2}{2}} \left( \dot{\te_t} - q_t\left(R_{\te_t}^\top x\right) \right) \matrice{0 \\ 1}.
\end{equation} 
From $\aaa^* (c_1+c_2) = 0$, we obtain  $c_1 = -c_2$ because $\aaa^*$ has trivial kernel.  
Thus, 
\begin{equation}\label{eq:7w}
    b_1(t,x) = c_1(t,R_{\te_t} x) U_{\te_t} \matrice{1 \\ -1} = c_1(t,R_{\te_t} x)  \matrice{e^{-i\te_t/2} \\ - e^{i\te_t/2}},  \ \ \ \ \aaa c_1(t,x) =  \dfrac{1}{2} e^{-\frac{x^2}{2}} \left( \dot{\te_t} - q_t\left(R_{\te_t}^\top x  \right)\right).
\end{equation}

We now use \eqref{eq:7x}: $q_t\left(R_{\te_t}^\top x\right) = \dot{\te_t} x_1^2$. Hence  $c_1$ satisfies the equation 
\begin{equation}\label{eq:0g}
    \aaa c_1(t,x) = \dfrac{1 - x_1^2}{2} e^{-\frac{x^2}{2}}  \dot{\te_t}.
\end{equation}
From the condition $b_1 \in \ker(T_0)^\perp$ we deduce that $c_1(t,x_1,\cdot) \perp e^{-x^2_2/2}$ for every $(t,x_1)$. Therefore, $c_1$ is explicitly given by:
\begin{equation}\label{eq:0h}c_1(t,x)=   \dfrac{1-x_1^2}{2} \,x_2\,e^{-\frac{x^2}{2}} \dot{\te_t}.
\end{equation}
This yields the identity \eqref{eq:0j} for $b_1$. 

\medskip 

We now focus on $f_1$. It solves the transport equation~\eqref{transport:fm}:
\[    D_t f_1(t,x_1) = \dfrac{1}{2\sqrt{\pi}} \int_\R e^{-\frac{x_2^2}{2}} \matrice{e^{i\te_t/2} \\ - e^{-i\te_t/2}} \cdot \beta_1\left(t,R_{\te_t}^\top x\right) dx_2,
\]
where by \eqref{eq:4b} $\beta_1=-T_1b_1-T_2a_0$.
In view of~\eqref{def:Tj}, $T_2$ is carried by $\sigma_3$ and  we deduce from  \eqref{eq:1o} that
\begin{equation}
    -\matrice{e^{i\te_t/2} \\ - e^{-i\te_t/2}} \cdot \beta_1(t,x) = 2D_t  \big( c_1(t,R_{\te_t} x) \big) 
    =  2 \left( D_t + \dot{R_{\te_t}} x \cdot D_x \right) c_1(t,R_{\te_t} x). 
\end{equation}
Using \eqref{eq:7u}, we obtain:
\begin{equation}
   - \matrice{e^{i\te_t/2} \\ - e^{-i\te_t/2}} \cdot \beta_1(t,R_{\te_t}^\top x) = 2\left(D_t +  \dot{R_{\te_t}} R_{\te_t}^\top x \cdot D_x \right) c_1(t,x) =  2\left(D_t + \dot{\te_t} \matrice{x_2 \\ -x_1} \cdot D_x \right) c_1(t,x),
\end{equation}
hence the transport equation for $f_1$:
\begin{equation}\label{eq:2c}
    D_t f_1(t,x_1) 
     =   - \dfrac{1}{\sqrt{\pi}} \int_\R e^{-\frac{x_2^2}{2}} \left( D_tc_1(t,x) + \dot{\te_t} (x_2 D_{x_1} - x_1 D_{x_2}) c_1(t,x) \right) dx_2.
\end{equation}

Thanks to the explicit formula \eqref{eq:0h} for $c_1$, we have:
\begin{equation}
  \int_\R e^{-\frac{x_2^2}{2}} D_tc_1(t,x) dx_2 = \dfrac{1-x_1^2}{2} e^{-\frac{x_1^2}{2}} \int_\R x_2 e^{-\frac{x_2^2}{2}} dx_2 \cdot D_t \dot{\te_t}  = 0.
\end{equation}
We deduce from integrating \eqref{eq:2c} and using the condition $f_1(0,x_1) = 0$   that
\begin{align}
   f_1(t,x_1) & = - \dfrac{1}{\sqrt{\pi}} \int_0^t \dot{\te_s} \int_\R e^{-\frac{x_2^2}{2}} (x_2 \p_{x_1} - x_1 \p_{x_2}) c_1(s,x) \ dx_2 \ ds
   \\
\label{eq:7v}   & = - \dfrac{1}{\sqrt{\pi}} \int_0^t \dot{\te_s} \int_\R x_2 e^{-\frac{x_2^2}{2}} (\p_{x_1} - x_1 ) c_1(s,x) \ dx_2 \ ds
\end{align}
where we have performed  an integration by parts in $x_2$. 
We now compute the integrals that appear in \eqref{eq:7v} using \eqref{eq:0h}. The integral on the LHS corresponds to integrating an odd function, hence produces $0$. 
Regarding the one on the RHS, we observe 
\begin{equation}
    (\p_{x_1} - x_1 ) c_1(t,x)=  x_2 (x_1^3-2x_1) e^{-\frac{x^2}{2}} \,\dot \theta_t.
\end{equation}
Therefore, the RHS of \eqref{eq:7v} becomes:
\begin{align}
  \dfrac{1}{\sqrt{\pi}} \int_0^t \dot{\te_s} \int_\R x_2 e^{-\frac{x_2^2}{2}} (\p_{x_1} - x_1 ) c_1(s,x) dx_2 ds
   & =  (x_1^3-2x_1) e^{-\frac{x_1^2}{2}} \int_0^t \dot{\te}_s^2 ds \cdot \dfrac{1}{\sqrt{\pi}}  \int_\R x_2^2 e^{-x_2^2} dx_2 
   \\
 \label{eq:0d}  & = \dfrac{2x_1-x_1^3}{2} e^{-\frac{x_1^2}{2}}  \int_0^t \dot{\te}_s^2 ds.
\end{align}
Plugging \eqref{eq:0d} in \eqref{eq:7v}, we conclude that
\begin{equation}\label{eq:0f}
   f_1(t,x_1) =  \dfrac{ x_1^3 - 2x_1}{2} e^{-\frac{x_1^2}{2}}  \Theta_t , \ \ \ \text{where} \ \ \ \Theta_t = \int_0^t \dot{\te}_s^2 ds.
\end{equation}
This completes the proof of Lemma \ref{lem:1h}. 
\end{proof}

\begin{proof}[Proof of Theorem~\ref{thm:3}] We set $a^{(2)}=a_0+ \epsi^{1/2} a_1 + \epsi b_2$, with $a_0, a_1, b_2$ solutions of
\begin{equation}
    T_0 a_0 = 0, \ \ \ T_0 a_1 + T_1 a_0 = 0, \ \ \ T_0 b_2 + T_1 a_1 + T_2 a_0 = 0; 
\end{equation}
see \S\ref{sec:3.2}-\ref{sec:3.7} for their construction. Thanks to Lemma \ref{lem:1e}, we have:
\begin{equation}
    \left\| (\epsi D_t + H) W\big[a^{(2)}\big]_{y_t} \right\|_{L^2} \leq C \epsi^2 \left( \big\| D_t b_2 \big\|_{L^2} + \big\| \lr{x}^4 a_0 \big\|_{L^2} + \big\| \lr{x}^4 a_1 \big\|_{L^2} + \big\| \lr{x}^4 b_2 \big\|_{L^2} \right).
\end{equation}

From the explicit expression \eqref{eq:3k} for $a_0$, $\| \lr{x}^4 a_0 \|_{L^2}$ is uniformly bounded. From the explicit expression \eqref{eq:0j} for $a_1$, $\| \lr{x}^4 a_1 \|_{L^2}$ is bounded by $1+\Theta_t$. It remains to bound $\| \lr{x}^4 b_2 \|_{L^2}$ and $\| D_t b_2 \|_{L^2}$. By construction, recalling that $r_t = 1$:
\begin{equation}
    b_2(t,\cdot) = - L_{\te_t,1}^{-1} \big( T_1 a_1 + T_2 a_0 \big).
\end{equation}
The explicit expressions for $a_0$ and $a_1$ allow us to bound Schwartz-class seminorms of $T_1 a_1 + T_2 a_0$ by $1+\Theta_t$ (the term $\p_t\Theta_t = (\p_t \te_t)^2$ is uniformly bounded). Arguing as in \eqref{eq:2s}, 
we deduce that Schwartz-class seminorms of $b_2(t,\cdot)$ and $D_t b_2(t,\cdot)$ are bounded by $1+\Theta_t$. In particular:
\begin{equation}
    \big\| D_t b_2 \big\|_{L^2} + \big\| \lr{x}^4 b_2 \big\|_{L^2} \leq C(1+\Theta_t). 
\end{equation}
We deduce that
\begin{equation}
    \left\| (\epsi D_t + H) W\big[a^{(2)}\big]_{y_t} \right\|_{L^2} \leq C \epsi^2 \left(1 +\Theta_t \right).
\end{equation}

We note that at $t=0$, $\Psi_t$ and $W\big[a^{(2)}\big]_{y_t}$ coincide up to $\Or_{L^2}(\epsi^{1/2})$. Thus, applying Lemma \ref{lem:1a}, we conclude that
\begin{equation}
    \left\| \Psi_t - W\big[a^{(2)}\big]_{y_t} \right\|_{L^2} \leq C \epsi^{1/2} +C \epsi t \left(1 +\Theta_t \right).
\end{equation}
This completes the proof of Theorem \ref{thm:3}.
\end{proof}

\subsection{Geometric setup.} \label{sec:4.2}
We prove here the  geometric facts stated above. First, if $\Gamma$ is a nodal set, then we can find a function $\kappa$ satisfying \eqref{eq:7s} with $\kappa^{-1}(0) = \Gamma$. %While this preserves the topological features of the experiment, changing $\kappa$ is a nearly impossible physical procedure: it corresponds to drastic modifications of the background structure.

\begin{lemm}\label{lem:1f} If $\Gamma = \tkappa^{-1}(0)$ for a function $\tkappa \in C^\infty_b(\R^2)$  satisfying the transversality condition \eqref{eq:7z}, then we can find $\kappa \in C^\infty_b(\R^2)$ satisfying \eqref{eq:7s} such that
$\Gamma = \kappa^{-1}(0)$. 
\end{lemm} 

\begin{proof}[Proof of Lemma \ref{lem:1f}] Without loss of generalities, we may assume that $|\nabla \tkappa(y)| = 1$ along $\Gamma$. We aim to construct $\rho \in C^\infty_b(\R^2)$ with $|\tkappa \rho|_\infty < 1$ such that if 
\begin{equation}\label{eq:7t}
    \kappa = \tkappa - \rho \dfrac{\tkappa^2}{2}=\tkappa \left(1-\frac { \tkappa\rho}2\right)
\end{equation}
then $\kappa$ satisfies \eqref{eq:7s}. Under the condition $|\kappa \rho|_\infty < 1$, $\kappa^{-1}(0) = \tkappa^{-1}(0) =\Gamma$. Moreover,
\begin{equation}
    \nabla \kappa = \nabla\tkappa \left(1 -  \rho\tkappa\right) - \dfrac{\tkappa^2}{2} \nabla\rho ;
\end{equation}
hence if $y\in \Gamma$ then $\nabla \kappa(y) = \nabla\tkappa(y)$. Also
\begin{equation}
    \nabla^2 \kappa = \nabla^2\tkappa (1-\rho\tkappa)
    - \rho  \nabla \tkappa \nabla\tkappa^\top
  - \tkappa \nabla \rho \nabla \kappa^\top - \tkappa  \nabla \tkappa \nabla\rho^\top     -  \dfrac{\tkappa^2}{2} \nabla^2\rho. 
\end{equation}
So, if $y\in\Gamma$ then $\nabla^2 \kappa (y) = \nabla^2\tkappa(y) -\rho(y)  \nabla \tkappa(y) \nabla\tkappa(y)^\top$. We deduce that for $y\in\Gamma$, 
\begin{align}
    \lr{\nabla \kappa(y), \nabla^2\kappa(y) \nabla \kappa(y) } & = \lr{\nabla \tkappa(y), \nabla^2\tkappa (y) \nabla \tkappa(y) } -\rho(y) \lr{\nabla \tkappa(y), \nabla \tkappa(y) \nabla\tkappa(y)^\top \nabla \tkappa(y) } 
    \\
    & = \lr{\nabla \tkappa(y), \nabla^2 \tkappa(y) \nabla \tkappa(y) } - \rho(y).
\end{align}
We now pick $\trho \in C^\infty(\R^2)$, such that  $\trho(y) = \lr{\nabla \tkappa(y), \nabla^2 \tkappa(y) \nabla \tkappa(y) }$ for $y \in \Gamma$. Then, with
\begin{equation}
   \rho(y) = \dfrac{\trho(y)}{1 + \trho(y)^2 \tkappa(y)^2}
\end{equation}
we still have $\rho(y) = \lr{\nabla \tkappa(y), \nabla^2 \tkappa(y) \nabla \tkappa(y) }$ for $y \in \Gamma$; $\rho \in C^\infty_b(\R^2)$; and finally,
\begin{equation}
    |\rho\tkappa| = \dfrac{|\trho \tkappa|}{1 + \trho^2 \tkappa^2} \leq \dfrac{1}{2}. 
\end{equation}
The function $\kappa$ given by \eqref{eq:7t} now satisfies the requirements of the lemma. Indeed, by construction we have for $y \in \Gamma$:
\begin{equation}\label{eq:0v}
\big|\nabla \kappa(y)\big| = \big|\nabla\tkappa(y)\big| = 1, \ \ \ \ 
    \lr{\nabla \kappa(y), \nabla^2\kappa(y) \nabla \kappa(y) } = 0.
\end{equation}
We can then write $|\nabla \kappa|^2 = 1 + \alpha \kappa$ for some smooth function $\alpha$. Taking the gradient on both sides produces the identity:
\begin{equation}
    2 \nabla^2 \kappa \cdot \nabla \kappa = \alpha \nabla \kappa+ \kappa \nabla \alpha.
\end{equation}
In particular, pairing with $\nabla \kappa^\perp$ gives
\begin{equation}
    2 \left \langle \nabla \kappa^\perp, \nabla^2 \kappa \cdot \nabla  \kappa\right \rangle = \kappa \left \langle  \nabla \kappa^\perp, \nabla \alpha \right \rangle.
\end{equation}
Specializing at $y \in \Gamma$ produces 
\begin{equation}
    \left \langle \nabla \kappa(y)^\perp, \nabla^2 \kappa(y) \cdot \nabla  \kappa(y)\right \rangle = 0,
\end{equation}
which together with the second identity of \eqref{eq:0v} yields $\nabla^2\kappa(y) \nabla \kappa(y) = 0$ when $y \in \Gamma$. \end{proof}

We now prove the useful relation~\eqref{eq:7x}.

\begin{proof}[Proof of \eqref{eq:7x}] We recall that $R_{\te_t}^\top e_1 = -\dot{y_t}=-\nabla \kappa(y_t)^\perp$ and $R_{\te_t}^\top e_2 = -\dot{y_t}^\perp=\nabla \kappa(y_t)$. Therefore, proving \eqref{eq:7x} boils down to showing
\begin{equation}\label{eq:7q}
    \lr{\dot{y_t}, \nabla^2 \kappa(y_t) \dot{y_t}} = \dot{\te_t}, \ \ \ \ \lr{\dot{y_t}, \nabla^2 \kappa(y_t) \dot{y_t}^\perp} = 0, \ \ \ \ \lr{\dot{y_t}^\perp, \nabla^2 \kappa(y_t) \dot{y_t}^\perp} = 0. 
\end{equation}
The last two identities are direct consequences of $\nabla^2\kappa(y) \nabla \kappa(y) = 0$ for $y \in \kappa^{-1}(0)$.
For the first identity in \eqref{eq:7q}, we note that
\begin{equation}
 \systeme{
 \cos(\te_t) = -\lr{\dot{y_t}, e_1} = \lr{\nabla \kappa(y_t),e_2} \\  
 \sin(\te_t) = -\lr{\dot{y_t}, e_2} = -\lr{\nabla \kappa(y_t),e_1}}.
\end{equation}
Taking time-derivatives, and the identity $\sin(\te_t)e_2 + \cos(\te_t) e_1 = -\dot{y_t}$, we deduce that
\begin{equation}
\systeme{     \dot{\te_t} \sin(\te_t) = -\lr{\nabla^2 \kappa(y_t) \dot{y_t},e_2} \\     \dot{\te_t} \cos(\te_t) = -\lr{\nabla^2 \kappa(y_t) \dot{y_t},e_1}} \ \ \ \Rightarrow \ \ \ \dot{\te_t} = \lr{\nabla^2 \kappa(y_t) \dot{y_t},\dot{y_t}}. 
\end{equation}
 This completes the proof of \eqref{eq:7x}. \end{proof}

\bibliographystyle{amsxport}
\bibliography{edgemode}

\end{document}